\documentclass[12pt,notitlepage]{article} 
\usepackage{times} 
\usepackage{cite}

\usepackage{graphicx} 
\usepackage{amssymb} 
\usepackage{amsmath} 
\usepackage{amsthm} 

\usepackage{setspace}
\usepackage{capt-of}
\usepackage[super, sort&compress]{natbib}

\newcommand{\nc}{\newcommand} 
\nc{\ro}{\mathrm} 
\nc{\ca}{\mathcal} 
\nc{\A}{\ca{A}} 
\nc{\PP}{\ca{P}} 
\nc{\Z}{\ca{Z}} 
\nc{\C}{\ca{C}} 
\nc{\Com}{\mbox{\Large$\Phi$}} 
\nc{\Om}{\Omega} 
\nc{\keywords}[1]{\par\addvspace\baselineskip\noindent\textbf{\textit{Keywords---}} #1}
\newtheorem{thm}{Theorem} 
\newtheorem{supp}{Supposition} 
\begin{document} 

\begin{titlepage}

\begin{center}
\Large\textbf{A Projective Algebra for Ansatz: Resolving Wigner’s Puzzle and the Existence of External Realms} \\
\vspace{10mm}

{\large
Jonathan M. M. Hall\footnote{\textit{email:}  
jonathan.corescience@gmail.com \indent\,\,\,\,\textit{website:} http://drjonathanmmhallfrsa.wordpress.com} 
 \\
\vspace{-1mm}
\normalsize The University of Adelaide, Adelaide, South Australia 5005, 
  Australia}  

\end{center}
\vspace{5mm}
\setstretch{1.5}

\begin{abstract} 
Natural philosophy integrates scientific observation with abstract frameworks, often using a mathematical \textit{Ansatz} 
to hypothesise about physical phenomena.
Exploring the possibility of other universes, however, challenges assumptions that physical laws, like spacetime geometry, extend beyond our reality. 
This paper argues that mathematical abstractions, serving as a ‘telescope’ beyond physical constraints, enable such reasoning. 
Through a projective algebra formalism (Section 4), we model the mechanism of \textit{Ansatz}, 
abstractly describing physical objects. This yields a resolution to Wigner’s ‘unreasonable effectiveness’ 
via cardinality equivalence (Section 5) and clarifies terms like ‘evidence’ and ‘existence’ (Section 6) 
to align with the conventions used in physics. 
A Cantor-inspired paradox shows no universe can contain all mathematical abstractions 
(e.g., sets, numbers), as its power set exceeds it, necessitating an external 
abstract realm (Section 6.4). 
This logical necessity, which holds even in the context of alternative set theories like New Foundations, 
provides evidence for a minimal external universe as an abstract realm, supporting Mathematical Realism.
This result is not specific to the formalism, as long as we accept that the principles of set theory 
are mathematically valid. 
As abstract entities elude empirical detection, logical evidence is apt, guiding future science and 
philosophy research and fostering interdisciplinary inquiry.
\end{abstract} 


\end{titlepage}

\tableofcontents
 \newpage

\setstretch{1.45}

\vspace{-3mm}\section{Introduction} 
\label{sect:intro} 

The study of physics inherently
 requires both scientific observation and philosophy. 
The tenets of science, and its axioms of operation, are not themselves 
scientific statements, but philosophical statements. 
The metaphysical underpinning of science includes statements about the process of science  
itself, and the nature of both the philosophical and material 
objects involved in a scientific investigation. 
Historically, the profound 
philosophical insight precipitating the birth of physics was that 
scientific observations and philosophical constructs, such as logic 
and reasoning, could be married together in a way that allowed 
one to make predictions of observations in science based on theorems, 
and proofs in logic, as a branch of philosophy, rather than simply 
to collect data on phenomena without interpretation. 
This natural philosophy requires a philosophical 
`leap', in which one makes an assumption about what 
abstract framework applies most correctly. Such a leap, called \textit{Ansatz},
 is usually arrived 
at through inspiration and an integrated usage of faculties of the
 mind, rather than 
a programmatic application of certain axioms. 
Nevertheless, once a set of fundamental principles are decided upon, 
a subsequent programmatic approach allows 
enumeration of the details of the ensuing formalism for the purposes of 
such an application. 
It seems prudent to apply a programmatic approach 
to the notion of \textit{Ansatz} itself and to clarify its  
process metaphysically, 
in order to gain a deeper understanding of how it is used 
in practice in science; but first of all, let us begin with 
the inspiration. 

\vspace{-3mm}\section{A metaphysical approach} 
\label{sect:meta} 
 
In this work, a programme is laid out for addressing the philosophical  
mechanism of \textit{Ansatz}. In physics in general, 
a scientific prediction is made firstly 
by arriving at a principle, usually at least partly mathematical  
in nature. The mathematical formulation is then `guessed' to hold in  
particular physical situations. The key philosophical process involved  
is exactly this `projecting' or `matching' of the self-contained mathematical  
formulation with the supposed underlying principles of the universe.  
No proof is deemed possible outside the mathematical   
framework, since proof, as an abstract entity, is an inherent 
feature of a mathematical (and philosophical) viewpoint.  
Indeed, it is difficult to imagine what tools a proof-like  
verification in a non-mathematical context may use or require.  

This view of the operation of the 
scientific method may at first seem controversial,
but in fact is equivalent to more conventional statements of the scientific 
method from Galileo \cite{galilei1914dialogues} to Feynman 
\cite{feynman1994character}.
In the traditional method, the flow from making observations to the 
formation of 
hypotheses, the making of predictions and the gathering of experimental 
data forms the skeleton of the scientific method. 
Note that the method by which the hypothesis is formed is not explicitly 
defined as part of the operations of science in this skeleton framework. 
A careful defense of the details within the 
science method, including the process of 
proposing an hypothesis as an \emph{Ansatz}, entails 
an important area of study and represents an ambitious task to 
say the least. 
 For the purposes of this investigation, it suffices to state 
that the main results of the paper 
are not specifically contingent on the details of a special construction 
of scientific method, but represents an investigation into the general 
aspects of a practical pursuit of the scientific method as it currently stands. 

It is possible that the current lack of clarity in the philosophical 
mechanism involved in applying mathematical principles to the universe 
 has implications for further research  
in physics. For example, in fine-tuning problems of the Standard Model 
of particle interactions (such as for the  
mass of the Higgs boson \cite{DGH,Martin:1997ns} and  
the magnitude of the cosmological constant 
\cite{ArkaniHamed:1998rs})  
it has been speculated that the existence of multiple universes  
may alleviate the mystery surrounding them 
\cite{Wheeler,Schmidhuber:1999gw,Szabo:2004uy,Gasperini:2007zz,Greene:2011}, 
in that   
 a mechanism for obtaining the seemingly finely-tuned value  
of the quantity would no longer be required - it simply arises statistically.  
However, if such universes are causally disconnected,  
e.g. in disjoint `bubbles' in Linde's Chaotic Inflation framework 
\cite{Linde:1983gd,Linde:1986fc}, 
there is a great challenge in even  
demonstrating such universes' existence, and therefore draws into question  
the rather elaborate programme of postulating them.  

Contemporary philosophical discussions include early universe cosmology 
\cite{Smeenk2014122}, quantum entanglement \cite{Barrett2014168}, 
causality \cite{Hoefer2014128,Valente2014101,Healey2014156} and quantum gravity 
\cite{Ashtekar2014}.
Setting aside for the time being the use of approaches that constitute novel  
applications of known theories, such as the 
 exploitation of quantum  
entanglement to obtain information about the existence of other universes 
\cite{Tipler:2010ft},   
a more abstract and philosophical approach is 
postulated below.  
 
\vspace{-3mm}\subsection{The universality of mathematics}
\label{subsect:univ}

Outside our universe, one is at a loss to  
intuit exactly which physical principles continue to hold. For example,  
could one assume a Minkowski geometry, and a causality akin to our current  
understanding, to hold for other universes and the `spaces between', 
if indeed  
the universes are connected by some sort of spacetime? 
Such questions are perhaps too speculative to lead to any 
real progress at this time;  
however, if one takes the view of 
Mathematical Realism, which often underpins 
the practice  
of physics (as argued in Section \ref{sec:mr})  
and the tool of \textit{Ansatz}, one may identify the 
 \emph{mathematical principles} as principles that should hold in any  
physical situation - our universe, or any other. 
This viewpoint is more closely reminiscent of Level IV in Tegmark's 
taxonomy \cite{Tegmark:2003sr} of universes. 
 One may imagine that mathematical theorems and logical  
reasoning hold in all situations, and that all `universes' - a term in need of  
a careful definition to match closely with the sense used in the 
practice of physics - 
are subject to mathematical inquiry. In that case,
 mathematics, and indeed, our own  
reasoning, may act as a `telescope to beyond the universe' 
in exactly the situation where all other  
senses and tools are drawn into question. 

To achieve the goal of examining the process of the \textit{Ansatz} - 
of matching mathematical ideas to non-mathematical entities (or phenomena),  
one needs to be able to define non-mathematical objects  
abstractly, or mathematically. Of course, such entities that can be  
written down and manipulated are indeed not `non-abstract'. 
This is so, in the same  
way that, in daily speech, an object can be referred to only by making  
an abstraction (c.f. `this object', `what is \emph{meant} by this object'  
`what is \emph{meant} by the phrase `what is meant by this object' ', etc.).  
 This nesting feature is no real stumbling block to the 
consistent construction of a formalism, as one can simply 
identify it as an attribute of a particular class of abstractions - 
those representing non-mathematical objects. Thus, a rudimentary 
(but sufficient, in this respect) 
formulation of non-mathematical objects in a mathematical  
framework
 will form the starting point for a new and fairly general formalism. 
Indeed, it is the general nature of this approach that gives it power, 
since very little structure is assumed about the non-mathematical 
`real-world' objects - only basic assumptions about the way 
such objects are referred to and handled in an abstract formalism; 
therefore, theorems and outcomes obtained from such a formalism are 
less prone to contamination from artefacts of the details of the formalism 
itself. 
 
After developing a mathematics of non-mathematical objects, one can then  
apply it to a simple test case. Using the formalism, one could, in principle, 
derive a  
process by which an object is connected or related somehow to its description,  
using only the theorems and properties known to hold in the new framework.  
The formalism could then even be applied to the search for other 
possible universes, 
and  the development of a procedure to identify properties of such  
universes. In doing so, one could indeed make a genuine discovery 
so long as the phenomenological properties are not introduced `by hand', 
and the underlying constraints built into the formalism are followed 
precisely.  
This follows the ethos of physics, whereby an inspired
 principle (or principles) is  
followed, 
sometimes seemingly remote from a phenomenon being studied, 
but which nevertheless 
has profound implications which are neither introduced artificially, 
nor always perceived contemporaneously, and  
which ultimately guide the course of an inquiry or experiment. 

There is an additional  
motivation behind this programme, beyond addressing the  
mechanism of the \textit{Ansatz}, 
which is to attempt to clarify philosophically Wigner's   
puzzle \cite{Wigner}. 
It is a key result of this  
paper to identify this kind of `effectiveness' as a kind of fine-tuning  
problem, i.e. that it is simply a feature that naturally 
arises statistically from the structure of the new formalism.

\vspace{-3mm}\subsection{Evidence}
\label{subsect:evintro}

In the special situation where one uses mathematical constructs exclusively, 
the type of evidence required for a new discovery 
would also need to be mathematical in nature, and  
testing that it 
satisfies the necessary requirements to count  
as evidence (in the usual scientific sense) would be achieved by using    
 mathematical tools within the new framework. 
To explain how this might be done, consider that evidence is usually taken to 
mean an observation (or collection of observations)
 about the universe that supports 
the implications of a mathematical formulation prescribed by a particular theory. 
 Therefore, it is necessary to have a strict separation between objects 
 that are considered `real/existing in the universe', and those that 
 are true mathematical 
statements that may be applied or \emph{projected} (correctly or otherwise) 
onto the universe. 
This idea is made more precise in Section~\ref{subsect:ev}. 

Note that, even for evidence in the usual sense,  
any observations experienced by the scientist are indeed abstractions also. 
For example, in examining an object, 
photons reflected from its surface can interact in the eye 
to produce a signal in the brain, and the interpretation of such a signal 
is an object of an entirely different nature to that of the actual photons 
themselves - 
much observational data are, in fact, discarded, and most crucially,
 the observation is then fitted into an abstract framework 
constructed in the mind. 
As a result, evidence can be thought of as prescriptive i.e. like a recipe, 
contingent on a pre-supposed mental model of experience, rather than an 
axiom. Whether a `phenomenon' constitutes evidence for a theory relies on 
an abstract relationship between the phenomenon cast as an item drawn out 
abstractly from experience to the point of being nameable, and another 
abstract construct, `the theory'. The phenomena themselves do not become 
evidence until they are encapsulated in a framework and viewed in this 
special kind of abstract relationship. 
Additionally, note that this 
does not mean evidence does not exist in this viewpoint, quite the contrary. 
It means evidence exists in the most real possible sense that every other 
kind of physical object exists. But it means that evidence, along with all 
physical phenomena, are required to inherit these abstract relations in order 
to be operated on for the purposes of the scientific endeavour, where evidence \emph{means something} - 
i.e. it needs to be juxtaposed with a theory.

In a very proper sense, the more abstract is the 
more tangible to experience, and the more material  
is the more alien to experience. Therefore, 
it seems reasonable to suggest that a definition of 
evidence in familiar scientific settings may already be equivalent 
to evidence in an entirely mathematical framework, 
and that the distinction between 
the two may be a matter of convention. 

Upon clarifying the type of logical evidence that is admissible for inquiry of mathematical objects, 
we then use a newly-developed projective algebra as a logical tool, 
demonstrating through Cantor’s Paradox that no single universe can contain all such abstractions, 
as the power set of any universe exceeds its capacity. This logical necessity proves 
the existence of abstract entities outside our universe, 
which we interpret as residing in external universes that are \emph{not} physical spacetime continua, 
but realms of \emph{mathematical structure}. 
This offers a philosophically grounded argument that relies solely on logic, 
the only admissible evidence for extra-universal entities.

\vspace{5mm}
\textbf{Summary}\\
To summarise and encapsulate the final thesis and import of this work that will follow from the presentation of the formalism, we will find  
\emph{there is no logically consistent way to define a universe that includes all abstractions (such as mathematical objects), and that such a claim need only rely on logic to prove it}.  
That is, objects outside the present universe can be proved to exist using a self-consistency definition, and that constitutes the only evidence that would be allowed in this case. 

\vspace{-3mm}\section{Mathematical Realism} 
\label{sec:mr} 

\vspace{-3mm}\subsection{Historical background}

In the `world of ideals' (as developed from the notion of 
Plato's `universals' \cite{Soc},  rather than Berkeley's Idealism \cite{Berk1}) 
there are certain abstract objects (`labels' or `pointers'), 
which refer to material objects. The 
\textit{Ansatz} arises by guessing and then assuming a particular connection 
between those pointers and other abstract objects. 
This allows material objects to be entirely objective 
(to avoid a solipsism), 
but also entirely subservient in some sense, to abstract objects. 
An observer can only indirectly interact via interpretation.  
Thus, the abstract affects the abstract, and abstract the physical. 

One might argue contrarily, 
that entities existing in the abstract mind 
are altered via natural or material means \cite{Jackson}. 
Certain mental states are invoked upon interpretation of the empirical, 
regardless of the existence of patterns, which are abstractions - and 
that this would be true even if the universe were  
`unreasonable', i.e. not in general amenable to rational inquiry.  
However, it is the point of view expounded in this treatise, 
that it is not true that 
material objects can interact so directly with other material objects, 
but only indirectly, 
since the interactions themselves would otherwise need to be materials.  
Instead, an interaction is taken to be necessarily an abstract link, i.e. it has to follow some 
pattern, rule or law, even in the simple case that
 the interaction between the objects 
is simply that they are at rest with each other. 
That is, the notion of `interaction' is \emph{necessarily} abstract. 

It is this 
feature of positing only indirect relation among physical objects 
that takes the form of a view opposite to that of 
Epiphenomenalism \cite{Jackson,Huxley}. 
The difference in viewpoint may appear to hinge on the semantics of 
the terms `interaction' and `abstract', but the goal is simply to 
characterise the oft-proved successful methods of physics  
\emph{as already applied} in practice, through a particular choice in 
philosophy. That is, by simply identifying (and thus labelling) the salient 
features of the philosophy, an investigation into the more general 
(and philosophical) aspects of the practice of physics can be conducted 
without need to re-demonstrate their soundness.  
On the other hand, the discussion of the soundness of such a philosophy 
on other grounds constitutes a tangent topic, and 
 the presentation of a complete enumeration of 
various emendations or contrary viewpoints on the choice of philosophy 
represents an interesting though separate, topic. 

Consider, as an example, a material object that collides with another 
material object, causing a change in the configuration of the system 
that includes both the objects and their immediate environment. The 
process is simultaneously observed by a person who interprets the interaction. 
At the point of the collision, it is of interest to characterise the 
kind of entity (material, or notional/abstract) identifying the collision or 
interaction itself. 
In the above `reverse-Epiphenomenal' viewpoint, 
no matter the physical mechanism for the collision, it is impossible to 
define the interaction without necessarily constructing a notional or 
abstract component. 
The reason for this is that the interaction is \emph{not defined} until 
there exists in the universe a rule or law of physics which the objects 
follow upon collision (such as the electrostatic potentials that repel 
each other at the atomic level). The physical laws themselves are notional, 
and whether they can be incorporated as elements of the physical 
universe in a consistent manner is discussed below. 

As an alternative, Libet suggests an Epiphenomenal viewpoint - 
that the interaction occurs as a physical entity, and 
the abstract states we assign to the objects in our own 
minds cannot be affected 
by other abstractions, but only by physical processes (with a 
caveat that a ``conscious veto'' to act on a physical process 
is possible for an individual) \cite{Libet}. 
This ties into the notion of \emph{qualia} \cite{Jackson}, 
in that the distinct quality of experience among individuals is incomparable 
in that it cannot be transported to a shared medium, 
and efforts to generate the same feelings in others cannot be verified 
externally - the experience depending intricately 
on the physical makeup of a person generating the qualia. 

Historically, the Process Philosophy, as developed by Whitehead, 
 helped tame the issue of the mind-matter duality \cite{whitehead2010process}. 
By characterising both the physical and abstract phenomena as derivative 
entities known as `occasions of experience', the processes 
and interactions themselves are promoted as ontologically fundamental. 
In the case of the above example, the type of entity identified 
as the collision is quite different in its properties to what is 
normally thought of as an object, and is in this respect 
concordant with reverse-Epiphenomenal view. 

From a Naturalistic monist viewpoint, one views all states of abstraction 
as emergent phenomena of the physical states of objects, including the brain, 
and its inherent mechanisms \cite{PAPQ:PAPQ053,mi2007naturalized}. 
`Imaginary' abstractions, i.e. those which do 
not directly occur via external stimulus or other input from the physical 
nature outside the perceiving individual, such as mathematical 
curiosities, are simply deemed `not existing', as they do not feature directly 
in the interactions of physical objects, and are safely ignorable 
in a construction of the workings of a physical system (except for those 
with a solipsistic viewpoint) \cite{ross2013scientific}. 

Of these contrasting views, the latter Naturalistic view is the most 
closely related to the metaphysical philosophy used in this article 
as a basis for the practice of science. 
Although a Functionalist framework is chosen (see below), Lewis argues  
that Functionalism is not complete as a metaphysic, and entirely compatible 
with versions of Physicalism (those that do not claim that all 
mental states are physical in nature), among other metaphysics 
\cite{Lewis1972-LEWPAT}. 
The Naturalistic and the Mathematical Realism viewpoints can be compared 
directly by bringing them into a consistent semantic, as follows. 
The Naturalistic viewpoint advocates distinct ontological status to 
material objects, and abstractions as non-existent artefacts - which 
nevertheless are useful in the construction of more sophisticated 
abstractions that more closely mirror the actual workings of the universe, 
for the practice of science. 

In the Mathematical Realism viewpoint, the same understanding of material 
objects may remain, and only a subtle difference occurring in the status 
of the abstractions, which are considered existing, but separate. 
Clearly, these two metaphysics may be operationally identical, and that the 
labelling of the space inhabited by abstractions as non-existing or existing 
could be a semantic one, since the properties that the abstractions follow 
(as far as they are defined) are not in contention. 
However, the reason for the strict separation of the abstractions into 
their own, independent `box' will become apparent in Section~\ref{subsec:univ} 
where it is discovered that a  
universe that encompasses both the abstract and physical objects, 
in the Naturalistic view, cannot be defined in a consistent way. 
This contrasts Field's Nominalism \cite{Field} which challenges such realism, arguing no need for 
distinct existence for mathematical entities, and for a reformulation of science to bypass 
these constructions, positing that they are not ontologically real.
As a counter-argument, the formalism presented in Sec~\ref{sect:form}, 
and the `relating theorem' introduced in Eq.~\ref{eq:ZIext} provides logical evidence 
for the necessity of the existence of mathematical entities.
It is therefore argued that it is suitable, in practicing science, to adopt the 
Mathematical Realism view (or a modified version of the Naturalism, 
cast as a semantically identical framework, such as suggested above).

\vspace{-3mm}\subsection{Contemporary developments}

The evolution of metaphysics from rudimentary Cartesian Dualism \cite{Desc} 
to that proposed by Bohm \cite{Bohm} demonstrates the usefulness of a mathematical 
viewpoint in clarifying and enriching philosophical ideas as they pertain 
to physics, and the universe as a whole. 
Abstract relationships are centralised, and underlying principles of 
matter, rather than a catalogue of the immediate properties,
 are interpreted to have the greater influence in accessing the  
fundamental nature of the physical world. 
The shift in perspective is that 
the simple and elegant descriptions of a physical system are those which  
incorporate its seemingly disparate features into an integrated whole. 
One then goes on to 
postulate a relationship between the physical object in question  
and the machinery of its abstract description. 

Similar philosophical views have found success in the field of neuroscience, 
such as the work of Dam\'{a}sio in characterising consciousness \cite{Damasio1, Damasio2}. 
Instead of treating the mind and body as separate entities, one postulates 
an integrated system, whereby mechanisms in the body, 
such as internal and external 
stimuli, result in neurological expressions such as emotion. 
Emotion and reason are thus brought together on equal footing, since both 
actions are the result of comparing and evaluating a variety of stimuli, 
including other emotions, to arrive at a response. 
That is, from a modern perspective, 
just as relationships between physical objects are  fundamental 
in characterising the intangible properties of their whole, 
it is the abstract relationships between faculties in the body and brain, 
such as interactions and stimuli, that characterise consciousness. 

The identification of phenomena with abstract descriptions, such as behaviour 
and interactions,  
was formalised in Putnam and Fodor's Functionalism \cite{Putnam,Block}. 
Functionalism provides the ability to consider 
indirect or `second order' explanations for the nature of objects. 
Unlike Physicalism, which identifies the nature of objects with the instances 
themselves that occur in the real world, Functionalism entails 
the generalisation of the objects in terms of their functional behaviour. 
These more general classes of object are identified by the features that 
all relevant instances of the object have in common, and so the nature of the 
 object 
becomes more ubiquitous, even if more abstract.  
This 
may be nothing more than a semantic shift, in cases where one is 
at liberty to allow the 
definitions of certain abstractions 
more scope as needed (such as pain or consciousness), 
so that they  more closely match their use in 
daily human endeavours. 
Further abstractions, 
such as `causes', an integral part of many areas of science, follow naturally. 

Colyvan, building on Quine and Putnam, argues that we must accept the existence of 
mathematical entities (e.g., numbers, sets) because they are indispensable to our 
best scientific theories (e.g., physics uses real numbers for space-time). 
This supports mathematical realism/platonism, as ontology follows from what our 
theories quantify (ie as 'bound variables') \cite{Colyvan}.

In this work, Mathematical Realism only diverges from Colyvan's Platonism by 
inclusion of an external realm, positing a separate universe for abstracts rather than 
embedding them in physical theories.
Taken on its own, Functionalism represents a deprecated metaphysic, 
insufficient for a complete account of internal states of a physical 
system \cite{Searle}, 
and is therefore commonly employed 
 simultaneously with another metaphysic (such as 
Physicalism). By not providing for a `real' or 
material existence independent of an object's functional behaviour, 
a bare Functional philosophy 
is not wholly suitable 
for describing the process of physics, which involves identifying 
material objects while projecting upon them an \textit{Ansatz} from some (abstract) theory. 

As an example of the shift in perspective that Functionalism brings, 
consider the following scenario of an abstract entity based on 
real-world observations, such as an emotion/state 
of affairs, etc.,  
whose cause is in want of identifying. Let this object be denoted 
as a `feature'. 
In this metaphysic, instead of the cause simply being identified as a  
specific phenomenon, or a mechanism 
based on the real world, 
the cause is characterised by the elements of 
an abstract object, $G$, which represents 
a collection of pointers to the relevant parts of the mechanism. 
For example, in the experience of pain, the mechanisms that cause 
it in an individual, the measurable features it produces, and the abstract 
notion of `pain' itself, which is $G$ in this case, are related 
in the following way. 
If we posit, for the moment, an abstract interpreting function, 
$i:R\rightarrow A$, relating the real world, $R$, to the world of ideals, $A$,
and similarly, a pointer function that goes in the opposite direction, 
$r:A\rightarrow R$, one can 
establish a relationship between the cause and the feature. 
 For a mechanism $m$, a set of features (visual characteristics, etc.) 
$F$, and an element $g_1\in G$, 
 define $m=r(g_1)$, which resides in $r(G)$, and 
 then $i(r(g_1))=i(m)\in F$. In this (fairly loose) symbolic description, 
the features and the cause become related via the statement $i(r(G))=F$.
The reverse-Epiphenomenal philosophy, akin to Interactionism, 
is to realise that the 
relation itself between the two entities is an abstract one, whose 
attributes it may benefit us to characterise. 

More recent constructions on the topic of the relationship between 
mathematical and real-world objects, such as the Inferential 
Conception approach \cite{Bueno2011-BUEAIC}, also predominantly take 
the view that certain mathematical structures play an important role 
in describing the behaviour of phenomena observed empirically \cite{LadymanFrench}. 
While superficially similar in approach, what marks this work as a different 
`take' on the issue of relating (or `mapping') mathematical and real-world objects 
 is the subtle shift in the underlying fundamental assumptions 
of ever being able to prescribe an explicitly denotative (let alone enumerable) description.  
This work propounds a description of a formalism that tries to do explicitly \emph{less} than 
provide a comprehensive description of objects (either mathematical or real-world). Instead, 
the emphasis is kept to addressing a specific problem, namely, the description of \emph{Ansatz}; 
and relevant information is encoded into the formalism 
that pertains only to it, which is developed, while leaving the theory free for more explicit descriptions to be encompassed later. 
In essence, the idea is simply to glean some possible aspect of the nature of the philosophical mechanism that underpins 
a crucial tool in physics (the \emph{Ansatz}); and because of this, one requires metaphysics as another important tool for the analysis of 
the methods of physics. 
As a result of this more general and less all-encompassing approach, 
a surprising array of outcomes are brought forth - in particular, new insight into the unreasonable effectiveness 
of mathematics \cite{Wigner}, fine-tuning problems, 
and other quandaries that lie directly in the purview of this kind of focused interrogation procedure. 

There is also a reciprocity involved in this kind of analysis - the problems faced by physicists help to shed light and guide investigation on certain points of philosophy 
that, up until now, have only been elucidated insofar as is important for addressing 
concerns in the ability to pin down definitively the fine (and sometimes mathematical) details of the objects in question. The view expounded in this work is, however, 
more conservatively along the lines of using a metaphysical description as a \emph{new tool} for physicists (or scientists at large) to use, thus opening hitherto less-accessible parts of the study of philosophy to a new demographic of thinkers by appealing to similar language and lines of thought, 
purely for addressing a small class of problems for which the formalism 
was explicitly designed. 
To summarise, this work principally represents a characterisation of the parts of 
metaphysics that pertain to physics \emph{as handled in practice} - questions that lie beyond the purview of traditional 
scientific inquiry but that nevertheless are important in understanding vital aspects of physics itself. 

Parallel investigations into the nature of objects, and the details 
of their interaction mechanisms, include 
the quantum mechanics of measurement in the Copenhagen and von-Neumann--Wigner 
interpretations \cite{stapp}, uncertainty relations and the 
(in)distinguishably of particles \cite{PhysRevA.86.062101}, and 
 the unreasonable effectiveness of mathematics \cite{gelfert}. 
A comprehensive study, however, on the 
links between the following practically-motivated formalism, as couched 
in the existing literature on the philosophy of science as a whole, 
represents a task that extends well beyond the scope of the present 
study, and will be left for future investigations into this topic.

\vspace{-3mm}\subsection{A brief caveat}

Before presenting the formalism, it is important to lay down a caveat 
pertaining to the interpretation of the results and outcomes of  
a formal approach to metaphysics of this kind, which may escape a casual 
reading of this work, unless explicitly stated. 
Critical feedback from researchers in the fields of metaphysics and the 
philosophy of physics are rightly wary of the concern that outcomes 
and theorems generated by such a formalism might be merely artefacts 
of the formalism itself, and the question of whether this mathematical structure 
reflects \emph{reality} may not be accessible to this formalism. 
As the author of the formalism, this concern is always at the forefront 
of one's mind, and great care is therefore taken in stating explicitly 
exactly what inferences are being made about the world, using this formalism. 

Firstly, the underlying principles of the formalism were carefully 
chosen in such a way that the definitions one 
is free to make reflect key observed structural properties, such as, for example, the `labelling 
principle' described in the next section, which encapsulates the recursive nature of 
referencing non-abstractions encapsulated within a (necessarily) abstract framework - e.g. 
`what is meant by what is meant by this object' (see Section~\ref{subsect:univ} above). 
In other words, the principles are chosen in such a way as to reflect the  
intrinsic structural nature of \emph{any} attempt at constructing a formal denotative language 
for describing objects (both non-abstract and abstract), reflecting reality in the specific 
area about to be interrogated, and then choosing to focus 
of these key observations. These choices, therefore, represent a potentially inescapable starting point for the construction of any such formalism from scratch, the starting point 
having to come from outside a strict formalism of that type. 

Secondly, and more importantly, it must be pointed out that the main outcomes 
of this work make very little assumption about real-world objects that could not 
possibly be accommodated in the formalism, thus keeping the formalism as general and 
`free of assumptions' as possible, and that the main conclusions drawn are on aspects (it can be non-controversially stated) that can be encapsulated naturally in such a scheme - 
namely, the properties of abstract objects (of the constructed type) \emph{themselves}. 
This should be borne in mind at all times in assessing the applicability of the outcomes of this work 
to distantly-related methods or areas of investigation. 

To cite a motivational example, the main intention of the construction of the formalism was not 
to seek a grandiose or unifying solution to all metaphysical problems, but to represent a genuine 
attempt to elucidate what could be considered to be a (troublingly) hitherto imprecisely-investigated 
aspect of the interface between science (focusing on physics) and philosophy (focusing on metaphysics): 
`what is meant by \emph{Ansatz}', and the parts of this process that specifically pertain to the relationship 
between the abstract and the `real world'. The idea that such a process, which is so crucially involved  
in the method of scientific inquiry (and most evident in physics), can be written down and then 
interrogated, using established techniques that are accessible due to the fact that the formalism is constructed to 
reflect how an \emph{Ansatz} is handled in practice, means that the formalism is, in essence, 
a way of computing metaphysical theorems. One need not agree with the underlying assumptions 
chosen for this instance of the formalism to concede that such an automation could, in principle, be immensely powerful if employed 
 in a variety of other metaphysical contexts that are specifically amenable to this kind of inquiry. 
It is the identification of these other contexts, that the formalism makes no claim to be able to do at this stage.

\vspace{-3mm}\section{Formalism} 
 \label{sect:form}

In this section, a new formalism is outlined in order to capture the  
essential features of the philosophical problem at hand.  
A set of very general \emph{abstraction operators} are defined, such that  
they may act upon each other in composition. By introducing another  
special kind of operation, the projection $\PP$, general objects may be  
constructed such that they, at the outset, obey the basic  
principles expected to hold for objects and attributes used in  
a recognisable context, such as in language.  
To avoid semantic trouble, when one is free to assume or assign  
a property in a given context, the choice made is that which  
is most closely aligned with `what is commonly understood' by a term.  
 Note that other definitions are (unless logically non-viable) completely  
acceptable also - it is simply a choice of convenience to try to align  
the concepts chosen to be investigated with those of a language,  
such as a spoken or written language; in fact,  
it is judicious to do so, given that any philosophical problems one may  
wish to address are usually initially cast in such a language.

\vspace{-3mm}\subsection{Projective algebra and abstraction classes} 
\label{subsect:proj} 
 
Cantor's Theorem \cite{Cantor} prohibits a  
consistent scheme classifying the space of all such abstract entities, 
(as echoed by Schmidhuber \cite{Schmidhuber:1999gw}). 
The abstractions considered here are limited to a set, $W$, 
of `real-world objects',  
representing a set of a specific type of object with certain  
(very general) properties. Very little mandatory structure for the 
 objects, $w$, inhabiting $W$ is assumed, and they may be represented  
by a set, a group or other more specific mathematical objects. 
Thus, one may define $W$ in a consistent fashion using an appropriate 
axiomisation, such as that of Zermelo-Fraenkel-Choice (ZFC) theory \cite{Zermelo},    
so long as none of the properties of the formalism is contravened. 

It’s important that the results of this work not rely on this particular construction, 
and represent general insights applicable regardless on which construction framework 
may be used. 
In this case, the salient feature relied on is that it is a mathematically consistent 
set theory that can allow a Russell paradox \cite{Russell}.
Other constructions that preclude a Russell paradox are also admissible; 
they simply represent a limiting case. 
As discussed later in Sec~\ref{subsec:den}, using New Foundations (NF) \cite{Quine} as a framework 
defines a smaller set of abstractions (with different cardinalities) 
and can allow a universal set to be consistently defined, 
and a set of abstractions harmonised neatly with a physical universe. 
In this case, while no external set to represent ‘left out’ abstractions in the universe 
is required, only one counter example is needed to generate this external set - the 
additional objects included in ZFC. 
In realism, the objects in a ZFC framework also need to be accommodated as real mathematical 
objects - and at the point of inclusion, a Russell Paradox results. 

In Section~\ref{subsubsect:gen}, the properties of the real-world  
objects $w$ are clarified. Note that, since $w$ must be handled in 
an abstract formalism, $w$ should not be thought of as a non-abstract object. 
It will be discovered that the treatment of $w$ must be the same as a `generalised object', 
described in Section~\ref{subsubsect:gen}, which automatically includes the 
most general recursive description of an object that receives interpretations from 
\textit{Ans\"{a}tze}. 
Some basic rules of composition are assumed, but the spaces mapped-into  
by doing so are simply definitions rather than theorems; the tone  
of this work is not to impose any more specific details on the framework  
than is required in order to fulfill the aforementioned goals, namely,  
the construction of a mathematical-like theory in order to address  
the mechanism of \textit{Ans\"{a}tze}. 
Other mathematical formulations for obtaining 
general information about a system, 
such as Deutsch's \emph{Constructor Theory} \cite{Deutsch:2012}, 
take a similar approach 
in determining suitable definitions for objects required for certain 
tasks in an inquiry.

\vspace{-3mm}\subsubsection{The labelling principle} 
\label{subsubsect:lp} 
Firstly, the projection operator, $\PP$, obeys what shall be known as  
the \textit{labelling principle}, which encodes the `nesting' feature 
inherent in referring to objects, identified in Section~\ref{subsect:univ}:  
\begin{equation} 
\label{eq:lp} 
\boxed{\PP\circ\PP = \PP.}
\end{equation} 
This encapsulates the aforementioned recursive nature of 
referencing non-abstractions such as `this object' being equivalent to `what is meant by this object' being equivalent to 
`what is meant by what is meant by this object', etc.
While a self-consistent mathematical framework may be like a hall of mirrors, 
this can simply be handled by encoding this property into the formalism.
 Since the projection operator non-trivially relates an abstract object to an object 
in the `real world' (and thus $\PP$ cannot be an identity operator), 
a direct consequence of the labelling principle is that  
$\PP$ has no inverse, $\PP^{-1}$.   
\begin{proof}
Assume $\PP^{-1}$ exists. Then:  
\begin{align*} 
\PP\circ\PP^{-1} &= \PP \\ 
&  1. \mbox{\quad Therefore, $\PP^{-1}$ DNE.}  
\end{align*} 
\end{proof} 
An equivalent argument follows for an operator $\PP^{-1}$ acting  
on the left of $\PP$. 

The projection operator may be applied to a real-world object $w$, 
and the resultant form, $\PP(w)$, constitutes a new object, inhabiting 
a different space from that of $w$. Firstly, the consequences of 
the lack of inverse of $\PP$ directly affect the projected space $\PP W$, 
which will be interpreted philosophically in the next section. 
Suffice to say, the judicious design of $\PP W$ lends itself to a 
particular  
view of abstractions, whereby very little information can be gained from 
an object in the real world \emph{directly}, as expostulated regarding 
the definition of evidence, in Section \ref{subsect:evintro}. 
 
\vspace{-3mm}\subsubsection{Abstractions} 

The notion of `abstraction' is codified by postulating a certain operator  
$\A$, 
which may act on objects residing in a space $W$, much like the projection 
operator. It will have, however, different properties to those of the 
projection operator $\PP$. 
Using the abstraction operator, one is able to go `up a level',  
($\A\circ\A(w)\neq\A(w)$), and establish new features of 
the object $w$.  
The sequential application of the abstraction operator creates a chain,   
 in a fashion analogous to cohomologies \cite{Bott}; 
however, the properties of the abstraction operators 
are more general than those of a functor. 

One may define the abstraction classes, $\Om^i$, as: 
\begin{align} 
1 &\in\Om^0,\\ 
\A &\in\Om^1,\\ 
\A\circ\A &\in \Om^2,\\ 
&\vdots \nonumber
\end{align} 
where $1$ is the identity operator; that is, where no information is added 
by an abstraction operator - the `trivial' abstraction. 
For the moment, the properties of the classes are no more extensive  
than, say, a collection of elements (the operators).  
The \emph{range} of $\A$, namely $\Om^1 \equiv \{\A_i\}$, is a class  
of any type of $\A$. (What is meant by the set symbols $\{$ $\}$ in the context 
of the new framework 
will be discussed in Section \ref{subsubsect:coll}.) 
It follows that, for $w\in W$, $w \in \Om^0(w) \in \Om^0(W)$.  

The sequential actions of the projection ($\PP$) and the abstraction ($\A$) 
operators do not  
cancel each other, and it can be shown that $\A\circ\PP(w)\neq w$: 
\begin{thm} 
\label{thm:cochains} 
$\A\circ\PP(w) \neq w$. 
\end{thm} 
\begin{proof}
Assume $\A\circ\PP(w) = w$. Then: 
\begin{align*} 
&\A\circ \PP\in\Omega^0 \\ 
&\Rightarrow \A\circ\PP = 1 \\ 
&\Rightarrow \A =\PP^{-1}, \,\ro{DNE.} 
\end{align*} 
\end{proof}\vspace{-3mm} 
Thus, $\A\circ\PP(w)\in\Om^1(\PP(w))$, for some $w\in\Om^0(w)$.  

Now, consider the complex of maps: 
\begin{align*} 
&\Om^0\rightarrow\Om^1\rightarrow\Om^2\cdots\\ 
&\downarrow\searrow^\chi\,\,\,\downarrow\qquad\downarrow\\ 
&\PP\,\Om^0\,\,\,\PP\,\Om^1\,\,\,\PP\,\Om^2\cdots 
\end{align*}
\begingroup\vspace*{-\baselineskip}
\captionof{figure}{The sequence of mappings from ranges associated with different abstraction levels, $\Om^i$, and their associated projections onto the real world.}
\vspace*{\baselineskip}
\endgroup
It is possible to design a function, $\chi\equiv\PP\circ\A$, which  
exists, and will be utilised in Section \ref{subsubsect:gen}.  
However, it is important to note that our constraints on $\PP$ do  
not allow the construction of a function   
$\varphi:\PP\,\Om^0\rightarrow\PP\,\Om^1$,  
or any other mapping between projections of abstraction classes.  
Mathematically, $\varphi$ would have the form $\PP\circ\A\circ\PP^{-1}$, which  
does not exist; but philosophically, it is supposed of  $\PP(w)$ that it 
should    
encode the behaviour of the actual real-world object \emph{meant} by $w$. 
One interprets the `non-mathematical' object, $\PP(w)\in\PP W$,  
knowing that the fact that it is necessarily an abstraction is already  
encoded in the behaviour of $\PP$ by construction (through 
the labelling principle).  
Note that the failure to construct a function $\varphi$ as a composite 
of abstraction operators and their inverses does not mean that such 
a mapping does not exist. However, the principal motivation for supposing 
the non-existence of such a function is to encode the features one expects 
in an abstract modelling of non-abstract objects. 

As mentioned in Section \ref{sec:mr}, 
this view of the general structure of abstraction is  
an opposite view to the metaphysic of Epiphenomenalism \cite{Jackson,Huxley}, 
in that, colloquially  
speaking, changes to `real-world' objects can only occur via some  
abstract state, and it does not make sense to set up a relationship  
between non-mathematical entities and insist that such a relationship  
must be non-abstract.  
 
Different instances of $w$ cannot be combined in general, but their   
abstractions can be compared by composition. 
The objects  
$\A(w_1)$ and $\A(w_2)$ can also  
be defined to be comparable, via use of the \emph{commutators}, in  
Section \ref{subsubsect:comm}.

In considering the properties of $\Om^0(W)$, one finds that:  
\begin{align} 
\Om^0(\Om^0(W)) &= \Om^0(W) \\ 
\Rightarrow \Om^0\circ\Om^0 &= \Om^0. 
\end{align} 
Generalising to higher abstraction classes, we find the \emph{level addition 
 property}: 
\begin{equation} 
\label{eq:la} 
\boxed{\Om^i\circ\Om^j = \Om^{i+j}.}
\end{equation} 
The non-uniqueness of $\A$ means that many abstract objects can  
describe an element of $W$. 
In general,  
$\A_i\circ\A_j(w)\neq \A_j\circ\A_i(w)$,  
so $\Om^1(\A_i(w))\neq\A_i(\Om^1(w))$, though both $\Om^1(\A_i(w))$  
and $\A_i(\Om^1(w))$ are in $\Om^2(w)$.  
The set $\Om^0$ includes the identity operator $1$, but also contains 
elements constructed from abstractions and other inverses, e.g. 
$\A_{L,i}^{-1}\circ\A_j(w)$, to be discussed in Section \ref{subsubsect:li}.

\vspace{-3mm}\subsubsection{Commutators} 
\label{subsubsect:comm} 

Define the commutator as an operator  
that takes the elements of the $i\,$th order abstraction space, acting  
on a real-world object $w_1$, to the same abstraction space acting on another  
real-world object $w_2$, $\Com^i_{W=\Om^0(W),1,2}: \Om^i(w_1)\rightarrow\Om^i(w_2)$.  
In this notation, 
the subscripts on the commutator symbol indicate the space inhabited 
by the objects whose abstractions are to be commuted, 
and the labels of the discarded 
and added objects, respectively. The superscript denotes the 
order of abstraction plus one, at which the commutation takes place. 
As a simple example, $\Com^1_{W,1,2}\A(w_1) = \A(w_2)$.  

In general, let: 
\begin{align} 
\Com^1_{W,i,j}\A(w_i) &= \A(w_j),\\ 
\Com^2_{\Om^1(w),i,j}\A_k\circ\A_i(w) &= \A_k\circ\A_j(w),\\ 
\mbox{and}\,\Com^{b+1}_{\Om^b(w),i,j} 
\underbrace{\A\circ\cdots\circ\A^{b}_i\circ\cdots\circ\A}_a(w) &=  
\A\circ\cdots\circ\A^{b}_j\circ\cdots\circ\A(w). 
\end{align} 
In order to construct the new object from the old object, 
one must successively apply `inverse' operations of the 
relevant abstractions to the left of the old object (as discussed 
in the next section), and 
rebuild the new object by re-applying the abstractions. 
This is not possible in general, where objects may include 
operators that have no inverse, such as the projection operator.

\vspace{-3mm}\subsubsection{The left inverse of the abstraction} 
\label{subsubsect:li} 
Define the left inverse: 
\begin{equation} 
\A_L^{-1}\circ\A = 1, 
\end{equation}  
or more generally, $\A_L^{-1}\circ\A(w) = w$.  
A right inverse is not assumed to exist in general, 
which will be important in establish certain kinds of properties in  
Section \ref{subsect:ext}.  

As a generalization, one can define a chain of negatively indexed  
abstraction classes $\Om^{-|i|}$. The level addition property can accommodate 
this scenario. The elements of $\Om^0$ are populated by objects of the 
form $\A_{L,i}^{-1}\circ \A_j$, \, or 
$\A_i\circ \A_{L,j}^{-1}\circ \A_k \circ \A_{L,n}^{-1}$,\, etc.   
That 
is, successive abstractions and inverses in any combination such that the 
resulting abstraction space is order zero. This includes the identity operator. 

By using the left inverse, it can be shown that the following theorem holds 
(which complements Theorem \ref{thm:cochains}), 
as a consequence of the  choice of the philosophical properties of $\PP$: 
\begin{thm} 
\label{thm:chains} 
$\PP\circ\A(w)\neq w$.  
\end{thm} 
\begin{proof}  
Assume $\Om^{-2}\neq\Om^{-1}$, and $\PP=\A_L^{-1}$. Then: 
\begin{align*} 
\PP\circ\PP(w) &= \A_L^{-1}\circ\A_L^{-1}(w)\\ 
&= \PP(w) \quad (\mbox{Eq.~(\ref{eq:lp})})\\ 
&= \A_L^{-1}(w)\\ 
\Rightarrow \A_L^{-1}\circ\A_L^{-1} &= \A_L^{-1} 
 \,\Rightarrow\Leftarrow 
\end{align*} 
\end{proof} \vspace{-3mm}
As a corollary, it also follows that: 
\begin{equation} 
\label{cor:chains} 
\PP\circ\A(w) \neq \A\circ\PP(w).  
\end{equation} 

In summary, this non-commutativity property of the abstraction operators in  
Eq.~(\ref{cor:chains}) is  
an important consequence of the reverse-Epiphenomenal  
philosophical motivation behind the labelling principle,  
and it will be the starting point for the construction of the generalised  
objects in Section \ref{subsubsect:gen}. 
 
\vspace{-3mm}\subsubsection{Auxiliary maps}

In order for a successful description of the relationships among different 
objects in general, 
 a definition of the mapping between objects of the form $\A\PP(w_1)$ and 
$\A\PP(w_2)$ is sought. 

Up until now, maps of the following types have been considered:\\
$\bullet\quad$ $w\rightarrow \A(w)$, where 
the notion of the map itself ($\rightarrow$) is an abstraction of $\A$ of the form $\A_{\ro{map}}\circ\A\in\Om^2(w)$; \\
$\bullet\quad$ $\A\circ\A(w)\rightarrow w$, where the map is now of the form 
$\A_{\ro{map}}\circ\A\circ\A\in\Om^3(w)$, and \\
$\bullet\quad$ $\A(w)\rightarrow w$, similar to the first case, except that 
the direction is reversed. 

By convention, let the definition of 
map be chosen such that the direction of the 
mapping is not important, but simply 
the relationship between the two objects. Therefore, the map is always 
taken such that it exists in a space $\Om^{i\geq 1}$, that is, we 
define it as the modulus of 
the map. Note that each of these maps is constructed as a composite of other 
abstractions, and shall be denoted \emph{literal maps}.

It cannot, in general, be attested that there is never to be a mapping 
between some $w_1$ and $w_2$ in $W$. However, one is free to \emph{define} 
to exist a non-composite type of map, denoted \emph{auxiliary maps}, 
which relate objects of the form $\A\PP(w_1)$ and $\A\PP(w_2)$, 
or $\A\PP\A_1(w)$ and $\A\PP\A_2(w)$, etc. The map itself exists 
in the space $\Om^1(w)$, that is, the application of what is meant by 
the map does not change the order of the objects to which it is applied. 

Using the concept of auxiliary maps, a generalised version of the 
commutator,
denoted $\hat{\Com}$, may be defined in a way not possible for the 
na\"{i}ve construction of the commutator: 
\begin{equation}
\hat{\Com}^2_{\Om^1(w),i,j}\A\circ\,\PP\circ\A_i(w) = \A\circ\,\PP\circ\A_j(w).
\end{equation}
The \emph{auxiliary commutator} will be important for the neat formulation of 
 the general condition for a specific kind of existence, 
in Sec~\ref{subsect:ext}. 

\vspace{-3mm}\subsubsection{The total set} 
One may define the \emph{total set} of an object $w$ as: 
\begin{align} 
S(w) &= \bigcup_{i=0}^\infty\Om^i(w),\\
S(\PP(w)) &= \bigcup_{i=0}^\infty\Om^i\circ\PP(w).
\end{align} 
The $\bigcup$ symbol denotes the range of the multiple objects, indexed 
by an integer e.g. $i$,    
over which the set should be specified. 
Here, one postulates a certain \emph{supposition of physics}, that  
$\PP S(w)$ spans at least $W$: 
\begin{supp} 
$\PP$ is surjective, i.e. $\forall w_i \in W, \exists \sigma_i\in S(w_i)$  
such that $\PP(\sigma_i) = w_i$.  
\end{supp} 
This condition represents the `working ethos' of the practice 
of physics. It is expected that there is some abstract description, 
however elaborate or verbose, to describe every real-world object.

\vspace{-3mm}\subsection{Relationships} 
\label{subsect:rel} 
 
In this section, the tools introduced in the preceding section  
will be used to define more general relationships between objects.  
In addition, a generalised object notation will be defined, and the  
nature of the  real-world objects $w$ will be  
clarified.  

\vspace{-3mm}\subsubsection{\textit{Ans\"{a}tze}} 
An \textit{Ansatz} is formed by adding a structure, or additional layer of 
abstraction, 
and imposing it on what one considers `the real world'.  
Cast in the new framework, this   
is simply the successive application of an abstraction and a projection  
operator upon some object: 
\begin{equation} 
\Z_i(w) = \PP\circ\A_i(w). 
\end{equation}  
As a consequence of the labelling principle of Eq.~(\ref{eq:lp}),  
the \textit{Ansatz} $\Z$ of an object $w$, and what is meant by the real-world 
object 
corresponding to $w$,  
cannot be related using a literal map, recalling that  
$\varphi:\PP\circ\A(w)\rightarrow\PP(w)$ does not exist in general.  
This simply means that there is no formulated 
procedure for generating the label of an  
object directly from the object itself - it is a free choice.  
The \textit{Ansatz} is akin to the statement: 
 `suppose this new label (with possible additional information) to  
be linked to the real object'. 
The notion of \textit{Ansatz}, particularly the special 
examples considered in Sec.~\ref{subsect:ext}, will be useful in understanding 
the formal structure of existence as it pertains to real-world objects.

\vspace{-3mm}\subsubsection{Collections and relationships} 
\label{subsubsect:coll} 
A \emph{collection} or \emph{set} of objects $\{w_i\}$ 
(indexed by integer $i$), 
in the formalism, 
is simply treated as 
an abstraction, $\A_\ro{set}$, used in conjunction with a commutator: 
\begin{equation} 
\{w_1,\ldots,w_N\} = \bigcup_{j=1}^N\A_\ro{set}\circ\Com^1_{W,i,j}w_i. 
\end{equation} 
By further imposing that there should be a relationship (other than the  
collection itself) among the objects $w_i$, the additional information  
is simply added by another abstraction, say, $\A'$,  
and what is meant be this particular relationship is simply: 
\begin{equation} 
r(w_i) = \PP\circ\,\A'\circ\bigcup_{j=1}^N\A_\ro{set}\circ\Com^1_{W,i,j}w_i. 
\end{equation} 
A relationship in general,  which is of the form 
$r = \A'\circ\Com^1\A\in\Om^2$, does not have  
to specify that there be a particular relationship among objects $w_i$.  

\emph{Example:} 
In identifying a `type' of object, such as all objects that satisfy  
a particular function or requisite, one means something slightly more  
abstract than a particular instance of an object itself.  
In order to produce a notion similar to the examples: 
`all chairs' or `all electrons',  
one must construct a relationship among a set of $w_i$'s, each of which  
is a set of, say, $n$ observations: $w_A,w_B,\ldots\in W'\subset W$. 

Let: 
\begin{align} 
w_A &= \bigcup_{j=1}^{N_A}\A_\ro{set}\circ\Com^1_{W,i,j}w_i^{(A)}, \\ 
w_B &= \bigcup_{j=1}^{N_B}\A_\ro{set}\circ\Com^1_{W,i,j}w_i^{(B)}, \quad\mbox{etc.} \\ 
w_i^{(A)},w_i^{(B)},\ldots&\in W. \\ 
\mbox{Then:}\,\,\{w_A,w_B,\ldots\} &= \bigcup_{j'=1}^{n}\bigcup_{j=1}^{N_{j'}}
\A_\ro{set}\circ\Com^1_{W',i',j'}\A_\ro{set}
\circ \Com^1_{W,i,j}w_i^{(i')}. 
\end{align} 
The relationship itself that constitutes the `type' thus takes the form: 
\begin{equation} 
\label{eq:type}
\boxed{r_\ro{type} = \A'\circ\bigcup_{j'=1}^{n}\bigcup_{j=1}^{N_{j'}}
\A_\ro{set}\circ\Com^1_{W',i',j'}\A_\ro{set}
\circ \Com^1_{W,i,j}w_i^{(i')} \in \Om^3(W)\,\, (= \Om^2(W')),}
\end{equation} 
for some $\A'$, and $W'=\Om^1(W)$\footnote{The form, $W'=\Om^1(W)$, 
makes sense, in that the notion of `being a subset' is a single-level 
abstraction, residing in $\Om^1$.}. 
This formula represents the notion 
of `types' of object in a fairly general fashion, in order 
to resemble as closely as possible the way in which objects are typically 
characterised and subsequently handled in the frameworks of 
language and thought.

\vspace{-3mm}\subsubsection{Generalised objects} 
\label{subsubsect:gen} 
 
Up until now, discussion of the nature of the real-world objects, 
$\{w_i\}\in W$,  
has been avoided. However, in order to incorporate them in the most  
general way into the framework of the abstraction algebra,  
one may posit that the real-world objects are simply a chain of  
successive projection or abstraction operators.  
In general, one can construct `sandwiches', such as: 
\begin{equation} 
\A_1\circ\cdots\circ\A_{i_1}\circ\PP\circ\A_2\circ\cdots. 
\end{equation} 
All objects considered in the framework thus far can be expressed in this  
form, noting that $\PP$(anything) $\in\Om^0\PP$(anything).  
Due to the corollary in Eq.~(\ref{cor:chains}),  
the projection operators cannot be `swapped'  
with any of the abstraction operators, and so the structure of the  
object is nontrivial.  
Let $c$ denote a generalised object, living in the space: 
\begin{equation} 
\label{eq:c} 
\boxed{c \in \Om^{i_1}\PP\,\Om^{i_2}\PP\cdots\PP\,\Om^{i_n}(W) \equiv \C^{i_1 i_2\ldots i_n}_W.}
\end{equation} 
The space $W$ here could stand for any other general space constructed in  
this manner, not necessarily the space inhabited by $c$ itself;  
thus Eq.~(\ref{eq:c}) is not necessarily recursive as it may initially appear.  
Because the `internal structure', so-to-speak, of $c$ contains a collection  
of a possible many abstractions, it may be expressed in terms of type. 
Here are two examples: 
\begin{align} 
\mbox{Let:}\,\,\,c^{(1)} &= \PP\{\PP\circ\A_1(w),\ldots,\PP\circ\A_n(w)\}  
\in \Om^0\PP\,\Om^1\PP\,\Om^1(w) = \C^{011}_w\\ 
&= \PP\bigcup_{j=1}^n\A_\ro{set}\circ\Com^2_{\Om^1,i,j}\PP\A_i(w).\\ 
\mbox{Or:}\,\,\,c^{(2)} &= \PP\{\PP(w_1),\ldots,\PP(w_n)\}  
\in \Om^0\PP\,\Om^1\PP\,\Om^0(W) = \C^{010}_W\\ 
&= \PP\bigcup_{j=1}^n\A_\ro{set}\circ\Com^1_{W,i,j}\PP(w_i), 
\end{align} 
where, in the first case, $\Om^0(w)\in\Om^0(W)$. 
 
Consider the behaviour of an \textit{Ansatz} $\Z=\PP\circ\A_Z$ acting on a  
generalised object $c\in\C^{i_1\ldots i_n}_W$: 
\begin{equation} 
\Z\circ c = \PP\circ\A_Z\circ\A_1\circ\cdots\circ\A_{i_1}\circ\PP\circ\cdots. 
\end{equation} 
$\Z$ maps $c$ into a space $\C^{0i_1+1\ldots i_n}_W$. If we define  
rank($c$)$= n$, then rank($\Z\circ c$)$= n+1$. Note that the rank  
of $\Z\in \C_c^{01}$ can also be read off easily: rank($\Z$)$= 2$.  
Objects of the form of $\Z$ are the principal rank $2$ \textit{Ans\"{a}tze}.  
Note that other rank $2$ objects besides $\Z$ exist, such as objects of the  
form $\PP\circ\A_1\circ\A_2$.  
 
A more general description of \textit{Ans\"{a}tze} also exists, analogously to the  
generalised objects. By constructing an object of the form: 
$\chi = \underbrace{\A\circ\cdots\circ\A}_{j_1}
\,\circ\,\,\PP\circ\underbrace{\A\circ\cdots\circ\A}_{j_2}\circ\PP\circ\cdots$, 
that is, for an object residing in a space $\C_c^{j_1\,j_2\ldots j_m}$, the composition  
$\chi\circ c$ lies in $\C_W^{j_1\ldots (j_m+i_1)\ldots i_n}$,  
which is of rank $n+m-1$.  
 
By convention, an \textit{Ansatz} must contain a projection operator. Therefore,  
there is no rank $1$ \textit{Ansatz}, 
and we arrive at our general definition of  
\textit{Ansatz}: 
\begin{equation} 
\mbox{\textit{Any object acting on $c$, with a rank $> 1$, is an Ansatz}.} 
\end{equation} 

In addition, there are no objects with rank $\leq 0$.  
\begin{proof}  
Let $\xi$ exist such that rank($\xi$)$\leq 0$, and $c\in\C^{i_1\ldots i_n}_W$.  
Then:  
\begin{align*} 
\mbox{rank($\xi\circ c$)} &= \mbox{rank($\xi$)+rank($c$)}-1 
< \mbox{rank($c$)}=n\\ 
\Rightarrow \xi\circ c &\in \C^{i'_1\ldots i'_{n-1}}_W \\
\Rightarrow \xi &\mbox{\,\,is of the form\,\,} 
X\circ\bigcup_{j=1}^{k}\PP^{-1}\circ 
\bigcup_{i=1}^{i_j} \A_{i,L}^{-1}, \,\mbox{where rank($X$)\,$\leq k$}\\
\Rightarrow \xi &\mbox{\,\,DNE, for any $X$.} 
\end{align*} 
\end{proof} 
 For example, in the case $k=1$, $X$ is a rank $1$ \textit{Ansatz}, and 
$\xi = X\circ\PP^{-1}\circ\bigcup_{i=1}^{i_1}\A_{i,L}^{-1}$, 
such that $\xi\circ c \in \C^{i_2\ldots i_n}_W\cong\C^{i'_1\ldots i'_{n-1}}_W$. The rank of $\xi\circ c$ is  $n-1$, 
and the rank of $\xi$ is therefore $0$. Because of the usage of the operation 
$\PP^{-1}$, such an object is inadmissible. 

We would like to use \textit{Ans\"{a}tze} to investigate the properties of  
generalised objects. However, there are a variety of properties in particular,  
discussion of which shall occupy the next section. The notion of `existence'  
is a key example that urgently requires clarification, and it will be  
found that such a property (and those similar to it),  
when treated as an abstraction, must have additional constraints.

\vspace{-3mm}\subsection{$I$-extantness} 
\label{subsect:ext} 
Firstly, one must make a careful distinction between what is meant by  
`existence' in the sense of mathematical objects, and in the sense  
of the `real world'. In the former case, one may assume that an object  
exists if it can be defined in a logically consistent manner.  
In the latter case, it is a nontrivial property of an object, which must  
be investigated on a case-by-case basis, and  
the alternative word `extantness' will be used for this in order to  
avoid confusion. The goal of the formalism is to \emph{relate} the two terms -  
that an object's extantness can be tested by appealing to the existence  
(in the mathematical sense) of some construction. This is a crucial 
point to understand in the formalism: that the notion of extantness 
of an object in the real world, by construction, will depend in some 
fashion upon mathematical existence. 
 
We begin by assuming that extantness is an inferred property of an object,  
and thus added by an \textit{Ansatz}. Define its abstraction, $\A_E$, such that  
an object $c = \PP\circ\A_E(w)$ is extant if such a construction
 exists; i.e. $c$ is extant  
if it can be written in this form (for any $w$).  
For $c = \PP\circ\A_E\circ\A_1\circ\cdots\in\C_W^{i_1\ldots i_n}$,   
the operator $\A_E$ must occur in the left-most position of  
all the abstractions in $c$. 
Clearly then, it \emph{must  
not necessarily be the case} that $\PP\circ\A_E(w)$ exists, if  
this abstraction is to be equivalent to how extantness  
(or existence in the conventional sense) is understood.  
 Otherwise, we have not correctly assigned the meaning of $\A_E$ 
to represent extantness properly. 
 
\emph{Example:} 
Consider the object $\PP\circ\A_E(1)$, where $1$ is the abstraction  
identity. $\A_E(1)=\A_E$ is the extantness itself, and $\PP\circ\A_E$ is  
`what is meant' by extantness, which is itself extant. It is the  
trivial extant object.   
 
This leads us to the first property of $\A_E$: that its right inverse,  
$\A_{E,R}^{-1}$, does not exist, as anticipated in Section \ref{subsubsect:li}.  
\begin{proof}  
Consider $c=\PP\circ\A\circ\cdots(w)$ such that it is not extant.  
Assume $\A_{E,R}^{-1}$  
exists also. Then:  
\begin{align*} 
c &= \PP\circ\A_E\circ\A_{E,R}^{-1}\circ\A\circ\cdots(w) \\ 
&= \PP\circ\A_E(w'), \,\mbox{where}\,w'\equiv\A_{E,R}^{-1}(w)\\ 
\Rightarrow &\,\mbox{$c$ is extant.}\,\Rightarrow\Leftarrow  
\end{align*} 
\end{proof} 

It is not necessary at this stage to suppose that the left inverse of  
$\A_E$ does not exist either; however, if that were the case, then  
$\A_E$ would share a property with $\PP$, in lacking an inverse.  
The two are unlike, however, in that $\A_E\circ\A_E\neq \A_E$.  
In other words, if we define $\A_E$ to live in a restricted class 
$\tilde{\Om}^1\subset\Om^1$, 
indicating 
 the additional constraint of lacking an inverse, then the level addition  
property of Eq.~(\ref{eq:la}) means  
$\A_E\circ\A_E\in\tilde{\Om}^2\not\equiv\tilde{\Om}^1$.  

A further consequence of the non-existence of $\A_{E,L}^{-1}$ is that  
the statement $\A_E(a) = \A_E(b)$ does not mean that $a=b$. One may interpret  
this as the fact that two abstractions may simply be labels for the same  
extant object. Note that the definition of the literal commutator requires the 
existence of an inverse of each abstraction operator that occurs 
in sequence to the left of the object being commuted, though that is not the 
case for auxiliary commutators. 

Recalling the supposition of physics, that  
$\PP S(w)$ spans at least $W$, a further clarification may now be added: 
\begin{supp} 
\label{supp:sopcor}
All extants have \textit{Ans\"{a}tze}, but not all elements of $\PP(W)$ or $\PP S(w)$  
are extants.  
\end{supp}  
From the point of view of Mathematical Realism, one would argue that 
 projected quantities, $\PP\circ\cdots$, are those which are `real' 
(and not dependent on their extantness), 
since such a definition of `real' would 
then encompass a larger variety of objects, including mathematical objects, 
regardless of their 
particular realisation in our universe. Such a choice for the 
word `real' seems to align best with the philosophy of Mathematical 
Realism; although this is purely a semantic choice. 
Nevertheless, it is still important to have a mechanism 
in the formalism to determine the extantness of an object. 
 
Although extantness has been singled out as a key property, a similar  
argument may be made for the `truth' of a statement, whose abstraction  
can be denoted as $\A_T$. Like extantness, the object $\PP\circ\A_T(w)$  
may not exist for every $w$, and the trivially 
true object is $\PP\circ\A_T(1)$. 
Let us label all properties of this sort, `$I$-extantness', since their  
enumeration in terms of common words is not of interest here.  
For any $I$-extant abstraction $\A_I$, we call $\tilde{\C}^{i_1\ldots i_n}_W$  
the restricted class of generalised objects, $c_I$.  
 
A formula is now derived, which is able to distinguish between  
objects that are $I$-extant and those that are not, by virtue of their  
 mathematical existence. Consider the case that $\PP\circ\A_I(w_1)$ exists, but  
$\PP\circ\A_I(w_2)$ does not.  Since the objects are 
characterised by the operators that appear in the generalised 
construction of Eq.~(\ref{eq:c}), $w_1$ must therefore 
contain a special property, $\A_{I'}$,  
that is not present in $w_2$. Unlike $\A_I$, it is not required that $\A_{I'}$  
occur in a particular spot in the list of abstractions that comprise $w_1$.  
Nor is there a restriction in the construction of an inverse, which would 
prevent a commutator notation being employed. 

Let $w_1$ be represented by a collection of objects defined by:  
$w_1=\{\A_{I'}\circ\A\circ\cdots,\A\circ\A_{I'}\circ\cdots,\mbox{etc.}\}$.   
 That is, $w_1$ takes the form of a set of generalised objects, $c$, but
 for the replacement of an operator, $\A$, with $\A_{I'}$. 
It is important to note that the particular $\A_{I'}$ that distinguishes 
$w_1$ from a non-extant object, such as $w_2$, is specific to $w_1$. 
For an object $c$ to be extant, it would have to include an abstraction 
$\A_{I'}^c$, specific to $c$; otherwise, any object related to $c$ in 
any way would also be extant, which would not reflect the behaviour 
expected of extant objects in the universe.  
In commutator notation, 
one would need to write out a geometric composition of the form:
\begin{align} 
w_1&= \A_{\ro{set}}\bigcup_{m=0}^{i_1-1}\Com^{i_1-m+1}_{\Om^{i_1-m}(w),m+1,I'}H_1\circ\PP\circ
\bigcup_{m'=0}^{i_2-1}\Com^{i_2-m'+1}_{\Om^{i_2-m'}(w),m'+1,I'}H_2\circ\cdots,\\
&= \A_{\ro{set}}\circ H_1\circ \bigcup_{p=2}^n\PP\bigcup_{m=0}^{i_p-1}\Com^{i_p-m+1}_{\Om^{i_p-m}(w),m+1,I'}H_p,
\end{align}
for $c=H_1\circ\PP \circ H_2\circ\cdots$. 
A more elegant formula may be defined simply in terms of $c$ itself,    
without the need of introducing new symbols, $H_1, \ldots, H_n$. 
One can achieve this using auxiliary commutators: 
\begin{align}
w_1&=\Big\{\hat{\Com}^{(\sum_{j=1}^n i_j)+1}_{\C_W^{i_1\ldots i_n},\,1,\,I'}c, 
\hat{\Com}^{\sum_{j=1}^n i_j}_{\C_W^{i_1-1\ldots i_n},\,2,\,I'}c,  
\ldots,\hat{\Com}^{(\sum_{j=1}^{n-1} i_j)+1}_{\C_W^{1\ldots i_n},\,i_1,\,I'}c,\ldots,  
\hat{\Com}^{(\sum_{j=1}^{n-2} i_j)+1}_{\C_W^{1\ldots i_n},\,i_2,\,I'}c,\dots\Big\}\\ 
 &= \bigcup_{m=0}^{i_1-1}\A_{\ro{set}}\circ\hat{\Com}^{(\sum_{j=1}^n i_j)+1-m}_{\C_W^{i_1-m\ldots i_n},\,m+1,\,I'}c 
\cup \bigcup_{m'=0}^{i_2-1}\A_{\ro{set}}\circ\hat{\Com}^{(\sum_{j=1}^{n-1} i_j)+1-m'}_{\C_W^{i_2-m'\ldots i_n},\,m'+1,\,I'}c \cup \cdots\\ 
&= \bigcup_{p=1}^n\bigcup_{m=0}^{i_p-1}\A_{\ro{set}}\circ
\hat{\Com}^{(\sum_{j=1}^{n-p+1}i_j)+1-m}_{\C_W^{i_1\ldots i_p-m\ldots i_n},\,m+1,\,I'}c.  
\end{align} 
It follows then, that a generalised object that is $I$-extant takes the form: 
\begin{equation} 
\label{eq:cIcond}
\boxed{c_I = \PP\circ\A_I\circ \bigcup_{p=1}^n\bigcup_{m=0}^{i_p-1}\A_{\ro{set}}\circ  
\hat{\Com}^{(\sum_{j=1}^{n-p+1} i_j)+1-m}_{\C_W^{i_1\ldots i_p-m\ldots i_n},\,m+1,\,I'}c,} 
\end{equation} 
where $c_I$ is of the form $\PP\circ\A_I(w_I)$. 
This is a powerful formula, as it represents the condition for $I$-extantness
 for a generalised object, $c$. 
Note that it would be just as correct to define $c_I$ as an element 
of a set characterised by the right-hand side (i.e. using `$\in$' instead 
of `$=$'), but because the notion of a `set', $\A_{\ro{set}}$ is simply an 
element of $\Om^1$ and it can be incorporated into the general form of 
$\C^{i_1\ldots i_n}_W$. 

One might wonder how to relate the properties of a proof (i.e. verifying 
the truth of a statement) with the 
existence of an abstraction $\A_{T'}$. 
In an example, consider the object representing the existence 
of truth, $w_T$. The validity of the `excluded middle' \cite{Aristotle} in this 
situation means that the proof is very simple:

\begin{proof}  
\vspace{-4mm}
\begin{align*} 
w_T &\Rightarrow w_T \\
\neg w_T &\Rightarrow (\neg \neg w_T) = w_T. 
\end{align*} 
\end{proof} 
Since $w_T$ is the statement of truth itself, i.e. $w_T = \PP\circ\A_T$, 
the inconsistency of $\PP\circ \A_T(w_T)$ means the inconsistency of 
$\PP\circ\A_T$. Such a statement is not true, by construction. 
One can now identify the abstraction, $\A_{T'}^{w_T}$ as being $\A_T$ itself. 
Thus, this exercise demonstrates that the proof of  
a statement has consequences for the abstract form of the statement, 
allowing one to identify more specific properties. 
Note that this does not, at this stage, provide us with extra proof methods, 
since there is no procedure, as yet, for acquiring knowledge of the 
form of an object's relevant $\A_{T'}$ in advance. The content of the proof 
must rely on standard means.

\vspace{-3mm}\subsection{Cardinality}

In the derivation of the general condition for an object $c$ to be 
$I$-extant (Eq.~(\ref{eq:cIcond})), one arrives at a set of elements. 
In this notation, the set is not intended to specify all the possibilities 
that each abstraction operator, $\A$, can take. Rather, 
the set can be thought of 
as being 
`the set of alterations from a general $c$ that encompass the required 
condition'. 

If one seeks an absolute measure of the `size' of the object, in terms 
of the overall possibilities, one may define a type of cardinality, $|c|$, 
in terms of the total possible number of abstractions.  
 Recalling Cantor's Theorem \cite{Cantor}, 
there is no consistent description of such a universal class. However, 
since the formalism accommodates the imposition of
 restrictions on the kind of objects that can be represented, let the 
number of possibilities for $\A$ be assumed consistently definable, 
and denote as $L$. $L$ need not be finite, nor even countable, however, 
it can be used to obtain formulae for the cardinality of an object. 

Define the number of abstractions, $\A$, in $c\in\C_W^{i_1\ldots i_n}$ 
as $\bar{n} \equiv \sum_{j=1}^n i_j$. Thus one finds that:
\begin{align}
\label{eq:cardc}
|c| &= L^{\bar{n}},\,\mbox{and}\\
|c_I| &= \bar{n}\,L^{\bar{n}-1}. 
\label{eq:cardcI}
\end{align}
The latter formula is simply a consequence of there being $\bar{n}$ 
possibilities for restricting one abstraction operator to be $\A_{I'}$. 
If one enforces $N$ restrictions on the set of $\A$'s,  then it follows that:  
\begin{equation}
|c_N| = \left(\frac{1}{2}(\bar{n}^2 - N^2)+1\right)L^{\bar{n}-N}. 
\end{equation}
This formula will become relevant in the next section.

\vspace{-3mm}\section{Unreasonable effectiveness} 
 \vspace{-2mm}

The goal is to use the general framework, described in Section \ref{sect:form}, 
to encapsulate the essence of describing phenomena using a theory  
(in the sense used in physics). Thus, the issue of Wigner's  
`unreasonable effectiveness' of mathematics \cite{Wigner} to describe 
the universe may be addressed by transporting the problem to a metaphysical 
context. There, the tools from philosophy, such as logic and proof theory, 
can be directed at the questions that involve not so much the 
behaviour of the universe, as the behaviour of \emph{descriptions} of the 
universe (i.e. the behaviour of the physics itself). 
It is important to be able to transport 
certain features of 
physics into a context where an analysis may take place, and such 
a context is, by definition, metaphysics. 

The notion of `effectiveness' is that, given a consistent set of 
phenomena, $v_i\in V$, one can extend $V$ to include more phenomena such that 
 \textit{Ans\"{a}tze} able to explain the phenomena 
satisfactorily can still be found. 
In this general context, what is meant by an `explanation' 
will be taken to be a relationship among the phenomena, $v_i$, in the form 
of abstractions. The essence of the mystery of the effectiveness of mathematics 
is not whether 
one can always `draw a box' around an arbitrary collection of 
objects, or that laws and principles (of any kind) are obeyed, but the ability
to identify particular principle(s) such that phenomena 
($v_1,v_2,\ldots$) are consequences of them; and that via the principles, 
the whole of $V$ may be obtained, indicating a more full explanation of the 
phenomena.  
That is, the phenomena 
are extant because of the truth of the underlying 
principles, rather than being identified 
`by hand' (which would hold no predictive power in the scientific sense). 
Note that the set $V$ may, in fact, 
only include a subset of the possible phenomena to 
discover in the universe, and so would represent a subset of the set, $W$, 
as discussed in Section \ref{sect:form}. 

Let $v_1,v_2\ldots$ have descriptions $\A_{v_1},\A_{v_2},\ldots
\in V\subset W$, which are extant. 
Let there be some principle 
(or even collection of principles with complicated inter-dependencies)  
described by the general object, $c_{\ro{princ}}$, such that each element of $V$ 
may be enumerated. It is our goal to investigate under what conditions 
the following statement holds:
\vspace{-4mm}
\begin{align}
c_{\ro{princ}} \,\mbox{is true} &\Rightarrow v_1,\,v_2,\ldots\,\mbox{are extant},\\
\mbox{i.e.}\,\,\Big[c_{\ro{princ}}=\PP\circ\A_{T}(w_{\ro{princ}})\Big] &\Rightarrow 
\Big[\A_{v_i} = \PP\circ\A_E(\A_{y_i})\in V\Big]. 
\label{eq:unr}
\end{align}
If there is a principle that implies such a statement, we wish to identify 
it, and investigate whether or not it is true. 

The circumstances of the truth of Eq.~(\ref{eq:unr}) depends on how 
$w_{\ro{princ}}$ is related to the phenomena, $v_i$. $w_{\ro{princ}}$ itself 
represents principle(s) whose truth is not added by hand in \textit{Ansatz} form. This does 
not mean that it is not true, since the form of $w_{\ro{princ}}$ is as
 yet unspecified. The most general way of relating $w_{\ro{princ}}$ and all 
$v_i$'s 
is to use the method of substituting abstractions into the formula 
for a generalised object, such as that used to derive the general condition 
of $I$-extantness in Eq.~(\ref{eq:cIcond}). In the same way that 
the set of all possible locations of $\A_{I'}$ in $c$ was considered, here, 
all possible combinations of locations of abstractions describing $v_i$ 
in $c$ must be considered, such that each $\A_{v_i}$ occurs 
\textit{at least once}. 
This formula can be developed inductively.
 
Consider initially only two phenomena, $v_1$ and $v_2$, with corresponding 
abstract descriptions defined as $\A_{v_1}$ and $\A_{v_2}$. For a generalised 
object, $c\in\C^{i_1\ldots i_n}_W$, one finds: 
\begin{equation}
w_{\ro{princ}}^{(N=2)} = \bigcup_{p=1}^n\bigcup_{\substack{m'=0 \\ m'\neq m}}^{i_p-1}
\bigcup_{m=0}^{i_p-1}
\A_{\ro{set}}\circ\hat{\Com}^{(\sum_{j=1}^{n-p+1}i_j)+1-m'}_{\C^{i_1\ldots i_p-m'\ldots i_n}_W,m'+1,v_2}
\hat{\Com}^{(\sum_{j=1}^{n-p+1}i_j)+1-m}_{\C^{i_1\ldots i_p-m\ldots i_n}_W,m+1,v_1} c. 
\end{equation}
In the case of $N$ phenomena, it is assumed that $N \leq i_p$:  
the number of abstractions available in the general formula for 
$c$ may be defined to be arbitrary large 
to accommodate the number of phenomena. One may make use of the following 
formula:
\begin{equation}
\bigcup_{\substack{m^{(N-1)}=0 \\ m^{(N-1)}\neq \mbox{\tiny{any other $m$'s}}}}^{i_p-1}
\cdots\bigcup_{\substack{m^{(2)}=0 \\ m^{(2)}\neq m^{(1)}}}^{i_p-1}\bigcup_{m^{(1)}=0}^{i_p-1} 
\qquad = \qquad \bigcup_{k=0}^{N-1}\bigcup_{\quad m^{(k,i_p)}\in[0,i_p-1]\backslash\bigcup_{\mu=0}^k\{m^{(\mu)}\}}.
\end{equation}
Here, $[0,i_p-1]$ is the closed interval from $0$ to $i_p-1$ in the set 
of integers, and for brevity, we define $\{m^{(0)}\}$ 
as the empty set: $\emptyset$. The most general form of $w_{\ro{princ}}$ 
may now be written as:
\begin{equation}
\boxed{
w^{\ro{G}}_{\ro{princ}} = \bigcup_{p=1}^n\bigcup_{k=0}^{N-1}
\!\!\!\!\!\!
\bigcup_{\quad m\in[0,i_p-1]\backslash\bigcup_{\mu=0}^k\{m^{(\mu)}\}}
\!\!\!\!\!\!\!\!\!\!\!\!\!\!\!\!
\A_{\ro{set}}\circ\hat{\Com}^{(\sum_{j=1}^{n-p+1}i_j)+1-m}_{\C^{i_1\ldots i_p-m\ldots i_n}_W,m+1,v_{k+1}}
c.}
\label{eq:WPG}
\end{equation}
In order for $c_{\ro{princ}}$ to be true, $w_{\ro{princ}}$ must contain 
information about the objects $y_i$, such that $\A_{v_i}=\PP\circ\A_E(\A_{y_i})$. 
Therefore, we seek only those elements of Eq.~(\ref{eq:WPG}) such 
that the phenomena $v_i$ take this form. This is a more restrictive set, 
as each abstraction of $y_i$ 
must have applied to it $\A_E$ and $\PP$ successively. 
There are only $n-1$ such occurrences of $\PP$ 
in $c$, so in making this restriction, we are free to choose:
\begin{align}
\label{eq:r1}
\bullet\quad N&\leq n-1, \\
\bullet\quad N &\leq i_p.
\label{eq:r2}
\end{align} 
The form of the more restricted version of $w_{\ro{princ}}$ 
is thus:
\begin{equation}
\boxed{w^{\ro{R}}_{\ro{princ}} = \bigcup_{k=0}^{N-1}
\!\!\!\!\!\!
\bigcup_{\quad p\in[2,n]\backslash\bigcup_{\pi=0}^k\{p^{(\pi)}\}}
\!\!\!\!\!\!\!\!\!\!\!\!\!\!\!\!
\A_{\ro{set}}\circ\hat{\Com}^{(\sum_{j=1}^{n-p+2}i_j)+1}_{\C^{i_1\ldots i_p\ldots i_n}_W,1,E}
\hat{\Com}^{\sum_{j=1}^{n-p+2}i_j}_{\C^{i_1\ldots i_p-1\ldots i_n}_W,2,y_{k+1}}c,}
\label{eq:WPR}
\end{equation}
where $\{p^{(0)}\}=\emptyset$. 

If $c_{\ro{princ}}$ can be constructed consistently, i.e. if it exists,
then the form of $w_{\ro{princ}}$ must be restricted to include an 
abstraction, $\A_{T'}^{w_{\ro{princ}}}$, 
that ensures the existence of $c_{\ro{princ}}$.  
This uses the same argument as in deriving Eq.~(\ref{eq:cIcond}), 
with $c_T = \PP\circ\A_T(w_T)$, 
and involves the union of Eq.~(\ref{eq:WPR}) with the object $w_T$. 
Thus, the condition 
under which Eq.~(\ref{eq:unr}) is true can now be written.  
\begin{thm} 
\label{thm:thcond}
The condition under which the principles 
of a theory describe certain phenomena takes the form: 
\begin{equation}
w_{\ro{princ}} \subseteq w^{\ro{R},T'}_{\ro{princ}} = \bigcup_{k=0}^{N-1}
\bigcup_{\quad p\in[2,n]\backslash\bigcup_{\pi=0}^k\{p^{(\pi)}\}}
\A_{\ro{set}}\circ\hat{\Com}^{(\sum_{j=1}^{n-p+2}i_j)+1}_{\C^{i_1\ldots i_p\ldots i_n}_W,1,E}
\hat{\Com}^{\sum_{j=1}^{n-p+2}i_j}_{\C^{i_1\ldots i_p-1\ldots i_n}_W,2,y_{k+1}}c\,\cup\, w_T.
\label{eq:thcond}
\end{equation}
\end{thm} 
\begin{proof}  
The statement of the theorem, 
that `$w_{\ro{princ}} \subseteq w^{\ro{R},T'}_{\ro{princ}}$ constitutes 
the condition for which Eq.~(\ref{eq:unr}) is true', 
is only fulfilled if the general form of $w_{\ro{princ}}$ in Eq.~(\ref{eq:WPG}) 
includes a description of extant phenomena explicitly, 
which takes the form shown in Eq.~(\ref{eq:WPR}). 
That is, one must show that 
$w^{\ro{R},T'}_{\ro{princ}}\subseteq w^{\ro{G},T'}_{\ro{princ}}$. 
This entails that the elements of  $w^{\ro{R},T'}_{\ro{princ}}$ and 
$w^{\ro{G},T'}_{\ro{princ}}$ are of the same form, differing only by use 
of a restriction. Therefore, in this case,  
the abstraction $\A_{T'}^{w_{\ro{princ}}}$
 is sufficient to ensure the truth of the elements in both sets. 
Note that the inclusion of $\A_{T'}^{w_{\ro{princ}}}$ takes the same form for 
both $w^{\ro{R},T'}_{\ro{princ}}$ and $w^{\ro{G},T'}_{\ro{princ}}$. 
Therefore, it is sufficient to show that 
$w^{\ro{R}}_{\ro{princ}}\subseteq w^{\ro{G}}_{\ro{princ}}$. 

Express Eq.~(\ref{eq:WPR}) in terms of abstractions, $\A_{v_i}$, 
recalling that 
$\PP\circ\A_E(\A_{y_i}) =$ \\
$\PP\circ\PP\circ\A_E(\A_{y_i}) = 
\PP\circ \A_{v_i}$: 
\begin{equation}
w^{\ro{R}}_{\ro{princ}} = \bigcup_{k=0}^{N-1}
\!\!\!\!\!\!
\bigcup_{\quad p\in[2,n]\backslash\bigcup_{\pi=0}^k\{p^{(\pi)}\}}
\!\!\!\!\!\!\!\!\!\!\!\!\!\!\!\!
\A_{\ro{set}}\circ\hat{\Com}^{(\sum_{j=1}^{n-p+2}i_j)+1}_{\C^{i_1\ldots i_p\ldots i_n}_W,1,v_{k+1}}c.
\label{eq:princp1}
\end{equation}
Choosing the value $m=0$ in $w_{\ro{princ}}^{\ro{G}}$ yields:
\begin{equation}
w_{\ro{princ},m=0}^{\ro{G},T'} =  \bigcup_{p=1}^n\bigcup_{k=0}^{N-1}
\A_{\ro{set}}\circ\hat{\Com}^{(\sum_{j=1}^{n-p+1}i_j)+1}_{\C^{i_1\ldots i_p\ldots i_n}_W,1,v_{k+1}}c.
\label{eq:princp2}
\end{equation}
The only difference between Eqs.~(\ref{eq:princp1}) and (\ref{eq:princp2}) 
is the choice of values of the iterator $p$.
To obtain $w^{\ro{R}}_{\ro{princ}}\subseteq w^{\ro{G}}_{\ro{princ}}$, 
it is sufficient to show that: 
\begin{equation}
[2,n]\backslash\bigcup_{\pi=0}^{k}\{p^{(\pi)}\} \subseteq [2,n] 
\quad \forall k\in[0,N-1]. 
\end{equation}
Recalling the restrictions of Eqs.~(\ref{eq:r1}) and (\ref{eq:r2}), 
 take $N \ll n$.   Now, 
$\bigcup_{\pi=1}^{k}\{p^{(\pi)}\}$ is a finite set of integers that is 
 a subset of $[2,n]$:
\begin{align*}
\bigcup_{\pi=1}^{k}\{p^{(\pi)}\} &\subseteq [2,n],\\
\mbox{where}\quad \{p^{(0)}\} &= \emptyset. \,\\
\therefore \bigcup_{\pi=0}^{k}\{p^{(\pi)}\} &\subseteq [2,n]\\
\Rightarrow w^{\ro{R},T'}_{\ro{princ}}&\subseteq w^{\ro{G},T'}_{\ro{princ}}. 
\end{align*}
\end{proof} 
Note that the fact that $w_{\ro{princ}}^{\ro{R},T'}$ is a more restrictive set
 than  $w_{\ro{princ}}^{\ro{G},T'}$ does not mean that it is `smaller' in the 
sense of cardinality. 
Assuming a sufficiently large value for $n$ to accommodate all $N+1$ restrictions, 
one finds that: 
\begin{align}
|w^{\ro{G},T'}_{\ro{princ}}| &= \left(\frac{1}{2}(\bar{n}^2 - (N+1)^2)+
1\right)L^{\bar{n}-(N+1)} = |w^{\ro{R},T'}_{\ro{princ}}| \\
&\Rightarrow w_{\ro{princ}}^{\ro{R},T'}\cong w_{\ro{princ}}^{\ro{G},T'}. 
\label{eq:equiv}
\end{align}

The above observation in Eq.~(\ref{eq:equiv})  
provides a possible explanation for the appearance of 
the `unreasonable effectiveness' of mathematics. The set 
of relationships among extant phenomena, $w_{\ro{princ}}^{\ro{R},T'}$, is not 
smaller, in any strict sense, 
than the general set of relationships among abstractions, 
$w_{\ro{princ}}^{\ro{G}.T'}$. 
The countability of sets of phenomena filtering into a more restrictive  
and still countable form, $w_{\ro{princ}}^{\ro{R},T'}$, 
combined with the formalism for describing 
non-mathematical objects in a mathematical way, constitutes the metaphysical 
explanation for the `unreasonable effectiveness' of mathematics. 
In other words, there is no unreasonableness at all, but it is 
simply a mathematical 
consequence of the countability of phenomena, and 
the abstract description of objects that are not innately 
abstract. 

Note, however, 
that it is not to be expected that objects in general are statistically 
likely to be extant. A comparison of the cardinalities in Eqs.~(\ref{eq:cardc}) 
and (\ref{eq:cardcI}) shows that $|c_I| < |c|$ if $L$ is uncountably infinite 
and $\bar{n}$ is at most countably infinite. Therefore, in the formalism, 
extantness itself should be seen as a special occurrence, i.e. that it is 
not `bound to happen' in general. 

This demonstrates the power of metaphysical tools, in the form 
of principles and proofs, to address key philosophical issues in physics. 
That is, 
the process employed here 
was not physics itself, but philosophical argumentation 
applied to the abstractions of objects used in the practice of physics.

\vspace{-3mm}\section{Evidence for universes}

\vspace{-3mm}\subsection{Defining evidence} 
\label{subsect:ev}

The definition of evidence relies on the connection between a set of 
phenomena (to be called `evidence'), and the principles of a theory 
that the evidence supports. Note that it is assumed here that the 
sense in which the phenomena support or 
demonstrate an abstraction, such as the theory, is the same sense in which 
a theory can be said to entail the extantness of the phenomena. The symmetry 
between the two arguments has not been proved, however, since it relies 
on the precise details of often-imprecisely defined linguistic devices.

The description of evidence, using the formalism of Section \ref{sect:form}, 
takes a similar form to that of the description of \textit{Ansatz} for phenomena  
in Eq.~(\ref{eq:unr}), except that the direction of the correspondence 
is reversed:
\begin{align}
\A_{v_1},\,\A_{v_2},\ldots\,\mbox{are extant} &\Rightarrow c_{\ro{princ}} 
\,\mbox{is true},\\
\mbox{i.e.}\,\,\Big[\A_{v_i} = \PP\circ\A_E(\A_{y_i})\in V\Big] &\Rightarrow 
\Big[c_{\ro{princ}}=\PP\circ\A_{T}(w_{\ro{princ}})\Big]. 
\label{eq:ev}
\end{align}
The left-hand side of Eq.~(\ref{eq:ev}) restricts the form of the object, 
$w_{\ro{princ}}$ (representing the set of principles) through 
 $w_{\ro{princ}} \subseteq w^{\ro{R}}_{\ro{princ}}$. 
For the form of $w_{\ro{princ}}$ to entail the right-hand side of 
Eq.~(\ref{eq:ev}), it must also include the abstraction, $\A_{T'}$. 
Thus, the condition under which Eq.~(\ref{eq:ev}) is true, where 
phenomena constitute evidence for a set of principles, is: 
\begin{equation}
\label{eq:evcond}
\Big[w_{\ro{princ}}\subseteq w^{\ro{R},T'}_{\ro{princ}}\Big]\quad (=c_{\ro{cond}}). 
\end{equation}
This is the same condition obtained for the examination of principles entailing 
extant phenomena, in Eq.~(\ref{eq:unr}).  

There is a duality between the two scenarios, which can be expressed in the 
following manner. 
If the condition of Eqs.~(\ref{eq:thcond}) and (\ref{eq:evcond}) is true, 
 then:
\begin{equation}
\Big[c_{\ro{princ}} = \PP\circ\A_T(w_{\ro{princ}})\Big]
\Leftrightarrow \Big[\A_{v_i}=\PP\circ\A_E(\A_{y_i})\Big].
\end{equation}
That is, relationship between principles and evidence is symmetrical, in a sense. The sort of phenomena entailed by a theory is of exactly the same nature 
as the sort of phenomena that constitutes evidence for such a theory. 
This leads one to postulate an object, $c_{\ro{D}} = \PP\circ\A_T(w_{\ro{D}})$, 
which represents 
this duality, and 
one may identify a \emph{Duality Theorem}, which takes the form:
\begin{equation}
\boxed{
w_{\ro{D}} = \Big\{\Big[\PP\circ\A_T(w_{\ro{cond}})\Big]
\Rightarrow\Big[\Big(c_{\ro{princ}}= \PP\circ\A_T(w_{\ro{princ}})\Big) 
\Leftrightarrow 
\Big(\A_{v_i}=\PP\circ\A_E(\A_{y_i})\Big)\Big]\Big\}.}
\end{equation}
The theorem is a consequence of the fact that 
the application of the restrictions acting upon $w_{\ro{princ}}$ 
commute with each other in the formalism. 

Note that, in attempting to clarify a  term ill-defined in colloquial 
usage, we have arrived at quite a strict definition of evidence:  
if $w_{\ro{princ}}$ is to constitute 
a set of principles describing the elements, $v_i$, it must at least take the 
form of a description based explicitly on all $v_i$ elements. 
Any part of $w_{\ro{princ}}$ that does not lie in 
$w_{\ro{princ}}^{\ro{R},T'}$ is not relevant for consideration as being 
supported by the evidence. 

\vspace{-3mm}\subsection{The relating theorem and the fundamental object} 
Since the $I$-extantness of some $c_I$ has been related to the mathematical 
existence of an object $\A_I(w)$, a primary question to investigate 
would be the $I$-extantness of the statement of \emph{this relation itself}. 
The statement of `the tying-in of the mathematical and non-mathematical 
objects' has certain properties that should deem the investigation of 
\textit{its own}  
$I$-extantness a nontrivial exercise. 

Denote the above statement, which is an \textit{Ansatz}, 
as $\Z_I(c_I)$, and let $c_I$ be $I$-extant. 
That is, for $c_I\equiv \PP\circ\A_I(w)$, $\A_I(w)$ exists; and let 
$\Z_I(c_I) \equiv \PP\circ\A_\Z\circ c_I$, for some $\A_\Z$. 
Recall that the assumed existence (in the mathematical sense) of the 
statement, $\Z_I(c_I)$, does not trivially entail $I$-extantness, 
under Supposition \ref{supp:sopcor}.  
To show that $\Z_I$ is $I$-extant, it is required that it can be put in 
the form:
\begin{equation}
\Z_I(c_I) = \PP\circ\A_I(w_\Z), 
\end{equation}
which implies that: 
\begin{equation}
\label{eq:ZI}
\PP\circ\A_\Z\circ c_I = \PP\circ\A_I(w_\Z).
\end{equation}
We would like to attempt to understand under what conditions this holds. 

Consider the scenario in which the $I$-extant form of $\Z_I$ does not 
exist. In this case, it is not possible to say that $\Z_I$ is not 
$I$-extant, since the statement relating existence and $I$-extantness 
has not been proved, and no information about $I$-extantness can be 
gained using this method. 
If, however, the $I$-extant form of $\Z_I$ does exist, 
then it is indeed certain that $\Z_I$ is $I$-extant. In other words, 
there is a logical subtlety that entails an `asymmetry': 
the demonstration of the existence of an object is enough to prove it, 
but the equivalent demonstration of its non-existence is not enough 
to disprove it, since the relied-upon postulate would then be undermined. 
Therefore, in this particular situation, 
unless a further logical restriction is found to be necessary to add 
in later versions of the formalism, it is sufficient to 
show that the $I$-extant form 
\emph{can} exist, for $\Z_I$ to be $I$-extant.
This is not true in general, due to Supposition \ref{supp:sopcor}, 
but holds for this special case. 
\begin{thm}
\label{eq:ZIext}
$\Z_I(c_I)$ is $I$-extant.
\end{thm}
The above theorem, denoted the \emph{relating theorem}, may be verified 
in proving Eq.~(\ref{eq:ZI}). It is enough to show that 
$\A_\Z\circ c_I = \A_I\circ w_\Z$ for any $c_I$, where there exists an $\A_\Z$ 
such that $\A_I$ obeys the property: $\A_{I,R}^{-1}$ 
DNE, which is, in our general framework,  
the only distinguishing feature of $\A_I$ at this point. 
The demonstration is as follows:
\begin{proof}  
Let $\Z_I = \PP\circ\A_\Z\circ c_I$ exist, such that:
\begin{align}
\A_\Z \circ c_I &= \A_\Z\circ\PP\circ\A_I\circ\bigcup_{p=1}^n\bigcup_{m=0}^{i_p-1}
\A_{\ro{set}}\circ\hat{\Com}^{(\sum_{j=1}^{n-p+1}i_j)+1-m}_{C_W^{i_1\ldots i_p-m\ldots i_n},\,m+1,\,I'}c \\
&= \A_I\circ w_\Z, 
\end{align}
for any $w_\Z$. Due to the labelling principle, this can only be true 
if $w_\Z\equiv c_I$ and $\A_I\equiv \A_\Z$; that is, the abstraction of 
$c_I$ (above) is an $I$-abstraction: $\A_\Z \in \tilde{\Om}^1(W)$.  
This is a valid choice, since the existence of $\A_{Z,R}^{-1}$ was not 
assumed. 
\end{proof}
  
Therefore, the form of $\Z_I$ is now known:
\begin{align}
\Z_I(c_I) &= \PP\circ\A_I(c_I)\\
&= \PP\A_I\circ\PP\A_I(w).
\end{align} 
In words, what has been discovered is that the \textit{Ansatz} of $I$-extantness 
is equivalent to the \textit{Ansatz} in the statement `the $I$-extantness of $c_I$ 
is related to existence'.  
That is, the operation associated with the statement: `it is $I$-extant' 
($\PP\circ\A_I(w)$), when applied twice, forms the statement 
`its $I$-extantness is related to existence' 
($\PP\A_I\circ\PP\A_I(w)$); and it is the \emph{same operation}. 
This need not be the case in general, and so it is a nontrivial result that:
\begin{equation}
\label{eq:ZIthm}
\boxed{\Z_I = \PP\circ\A_I.} 
\end{equation}
Note that Eq.~(\ref{eq:ZIthm}), in this case, is not a definition, 
but a \emph{theorem}, to be known as the \emph{correspondence corollary} 
to the relating theorem. 
  
In a sense, $\Z_I$ is the fundamental $I$-extant object, 
in that it is the most obvious starting point for the analysis 
of the existence of $I$-extant objects in general. 
It also constitutes the first example of an object \emph{demonstrated to 
exist in a universe} (though, a clarification of distinguishing 
different universes is still required, and investigated in Section 
\ref{subsec:univ}). 

Recall, in construction of `types' in Eq.~(\ref{eq:type}), that 
familiar notions such as `chair', or other such objects, are 
brought into a recognisable shape using this formula. 
Though the types may not appear more recognisable at face value, 
the properties of such a construction align more closely with what is 
meant phenomenologically by such objects. 
In a similar fashion, the type of $\Z_I$ can be established, to create 
a more full, complete, or `dressed' version of the object. 

$\Z_I$ is an example of a $c_I$. Since projected objects cannot be 
related directly, a type will be constructed from instances of $\A_I$ 
(of which $\A_E$ is one), and the dressed fundamental object will be 
a projection of the type. Using the same argument used in deriving 
Eq.~(\ref{eq:type}), the set of $n$ observed instances of $\A_I$ 
takes the following form (acting on some set $W'$):
\begin{equation}
\bigcup_{j'=1}^n\A_\ro{set}\circ\hat{\Com}^2_{\tilde{\Om}^1(W'),i',j'}\A_I^{i'}(W').
\end{equation}
If each observed instance may be identified as the set of a certain $N$  
characteristics residing in $W$, then our intermediate set $W'$ can be dropped, 
and we find:
\begin{equation}
\label{eq:Aprime}
\A_I\circ \bigcup_{j=1}^{N_{i'}}\A_\ro{set}\circ\hat{\Com}^1_{W,i,j}w_i^{(i')} \in \A_I^{i'}(\Om^1(W)).
\end{equation}
In this case, an instance of $\A_I^{i'}$ contains more information than just 
a set of characteristics, since it is also known that it is $I$-extant. 
Therefore, it contains an additional abstraction operator. 
What is meant by the `fundamental type' therefore takes the form:
\begin{equation}
\label{eq:fundtype}
\boxed{R_\Z\equiv\PP r_\Z =
\PP\,\A'\circ\bigcup_{j'=1}^n\bigcup_{j=1}^{N_{j'}}
\A_\ro{set}\circ\hat{\Com}^1_{W,i',j'}
\A_I\circ\A_\ro{set}\circ\hat{\Com}^1_{W,i,j}
w_i^{(i')}\in\PP\,\Om^2\circ\tilde{\Om}^1\circ\Om^1(W),}
\end{equation}
where $\A'$ is the abstraction of relationship. Note that $R_\Z$ is, 
in general, an element of $\PP\Om^4(W)$.

\vspace{-3mm}\subsection{Distinguishing universes}
\label{subsec:univ}

In this section, we address the issue of classifying universes by their 
properties in a general fashion. An attempt can then be made to 
identify features that distinguish universes from one another, 
and thus clarify the definition of our own universe in a way that 
is convenient in the practice of physics. 

Suppose the definition of 
a universe, $\ca{U}$, to be the `maximal' list of objects that have 
the same character, that is, obeying the same list of basic properties. 
The list need not necessarily be finite, as each collection of 
properties could, in principle, represent a collection of infinitely 
many objects themselves. 
In the language of the formalism developed so far, a 
formula 
may be constructed from a generalised object $c$ by ensuring 
that each element of $\ca{U}$ 
is related to the content of the underlying principles, $w_{\ro{princ}}$. 
This formula will be analogous to Eq.~(\ref{eq:WPG}). 
\begin{supp} 
\label{supp:univ}
$\exists\, \ca{U}$,
such that a list of underlying principles may be configured to be 
 enumerable as a countable set, 
$w_{\ro{princ}} = u_1\cup\cdots\cup u_N$, for $N$ elements. 
(It is not required that $u_1\cup\cdots\cup u_N\in \ca{U}$.)
A universe based on these principles is the object represented 
by the largest possible set of the form:
\begin{equation}
\label{eq:univ}
\ca{U} = \PP\circ\bigcup_{p=1}^n\bigcup_{k=0}^{N-1}
\!\!\!\!\!\!
\bigcup_{\quad m\in[0,i_p-1]\backslash\bigcup_{\mu=0}^k\{m^{(\mu)}\}}
\!\!\!\!\!\!\!\!\!\!\!\!\!\!\!\!
\A_{\ro{set}}\circ\hat{\Com}^{(\sum_{j=1}^{n-p+1}i_j)+1-m}
_{\C^{i_1\ldots i_p-m\ldots i_n}_W,m+1,u_{k+1}}c,
\end{equation}
with respect to a generalised object, $c$. 
\end{supp}
Note that extantness is a universal property, in that it can be 
defined in the formalism regardless of the universe in which it is 
extant. It may, but it is not required that it constitute one of the $N$ 
underlying principles of a universe. 

The distinction between different universes is largely convention, 
based on the most convenient definition in practising physics. 
One such convenience is the ability to arrive at a \emph{consistent} definition 
of the universe. This is not the case for a na\"{i}vely-defined universe 
  required to contain all possible objects, due to 
Cantor's Theorem \cite{Cantor}. 
In order to establish two universes as distinct, the 
following convention is adopted:
\begin{supp} 
\label{supp:dist}
Consider two consistently definable universes $\ca{U}$ and $\ca{U}'$. 
If the universe defined as the union, $X \equiv \ca{U}\cup \ca{U}'$ is 
inconsistent, 
then the universes $\ca{U}$  and $\ca{U}'$ are distinct. 
\end{supp}
Introducing a square-bracket notation, where $\ca{U}[\ldots]$ indicates 
that the underlying principles to be used in defining $\ca{U}$ are listed 
in $[\ldots]$, one may write $\ca{U} = \ca{U}[w_{\ro{princ}}]$ and 
$\ca{U}' = \ca{U}[w'_{\ro{princ}}]$. Define the following lists of principles 
 to be consistent:
\begin{align}
w_{\ro{princ}} &= u_1\cup \cdots u_{N-1},\\
w'_{\ro{princ}} &= u_N \cup u'_1\cup \cdots u'_{N'},
\end{align} 
but suppose the inclusion of both principles $u_{N-1}$ and $u_N$ to lead to 
inconsistency. It then follows that:
\begin{equation}
 X \equiv \ca{U}[u_1\cup \cdots u_{N-1}]\cup \ca{U}'[u_N \cup u'_1\cup 
\cdots u'_{N'}]\quad\mbox{is inconsistent.}
\end{equation}
\emph{Example:} The inconsistency of a set of principles can emerge 
in the combination of negation and recursion, as clearly demonstrated 
by G\"{o}del \cite{Godel} and Tarski \cite{Tarski}. 
If $u_{N-1}$ were to express a negation, such as `only contains 
elements that don't contain themselves', and $u_{N}$ were to enforce 
a recursion, such as `contains all elements', then Russell's paradox 
would result \cite{Russell}.

\vspace{-3mm}\subsection{Evidence for other universes}
\label{subsec:den}

The final denouement is to demonstrate that the fundamental type constitutes 
evidence for a universe distinct from our own. 
Consider $N$ abstract objects, $\A_{v_k}$, each of which represents 
an element of the fundamental 
type in Eq.~(\ref{eq:fundtype}). They may be expressed 
analogously to the operator 
$\A_I^{i'}$ defined in Eq.~(\ref{eq:Aprime}), 
for $\ca{N}$ sub-characteristics:
\begin{equation}
\A_{v_k} \equiv \PP\circ\A_I\circ \bigcup_{j=1}^{\ca{N}}\A_{\ro{set}}\circ\hat{\Com}^1
_{W,i,j}w_i^{(k)}.
\end{equation} 
Though the sub-characteristics themselves are not vital in this investigation, 
one can simply see that the objects are $I$-extant (by construction), 
by expressing them in the form:
\begin{equation}
\A_{v_k} = \PP\circ\A_I(\A_{y_k})\in V',
\end{equation}
where $V'$ denotes a set that contains at least all $N$ elements, $\A_{v_k}$. 
By determining the underlying principles describing $V'$, one may denote its 
maximal set as\, $\ca{U}'$. 
Since $V'$ is an abstract object existing as a subset of the objects
 that comprise our formalism, $F$, it follows that $\ca{U}'$ is the set of 
all abstracts, and cannot be consistently defined, due to Cantor's argument \cite{Cantor}. 
That is, by taking the maximal set of objects obeying this restriction, 
one arrives at a set containing itself, and all possible abstract objects, 
which is Cantor's universal set. 
This does not mean, however, that $V'$ itself is inconsistent;   
since $V' \subset F \subset \ca{U}'$, we are free to assume that $V'$ 
is constructed in such a way as to render it consistent, just as was 
assumed for $F$. 

If NF \cite{Quine} is used instead of ZFC as a framework, 
the maximal set $\ca{U}'$ can be defined consistently, and it isn't necessarily the case that  
$X \equiv \ca{U}\cup \ca{U}'$ is 
inconsistent, and the universes $\ca{U}$  and $\ca{U}'$ are \emph{not} mandatorily distinct in this case.
However, this does not fully align with Maddy's naturalism \cite{Maddy} avoiding external realms 
or Colyvan's holism \cite{Colyvan} of folding mathematical entities into the universe, 
as ZFC's validity in realism implies the existence of $\{x | x \notin x\}$ objects that 
would be left out of an NF construction. Essentially, specific objects have been \emph{left out} in order to accommodate 
a naturalist view, which have an existence to be accounted for, according to Mathematical Realism.
Akin with Hamkins' pluralism \cite{Hamkins}, we 
are then left with a robust consequence across different set theories. 
To show the need for an external universe, we only need to find one mathematically valid 
example, where, at the point of inclusion, a Russell Paradox results. 
Realism accepts ZFC as a valid framework, ensuring at least one inconsistent object necessitates and external  $\ca{U}'$.
Therefore, 
we are only left with proving that the 
full gamut of mathematical abstractions cannot be imported into $\ca{U}'$  
without generating an inconsistency.

Consider the physical universe, $\ca{U}$, 
which takes the form of Eq.~(\ref{eq:univ}), 
with the restriction that one of its underlying principles must be 
that all elements are extant. 
That is, let \,$\ca{U}$ only contain elements of the form $w = \PP\circ\A_E(c)
\in \ca{U}$. 
$\ca{U}$ may be still be 
consistently defined. 
Since $\ca{U}'$ is inconsistent, the definition of a composite universe 
$X \equiv \ca{U} \cup \ca{U}'$ is also inconsistent. \emph{Therefore, 
 by the convention established in Supposition \ref{supp:dist}, 
\,$\ca{U}$ and \,$\ca{U}'$ are distinct}. It also follows that $V'\subset \ca{U}'$ 
is also distinct from $\ca{U}$. 

Note that a statement that asserts the existence of an object in 
another universe simply based on a mathematical framework 
would be in danger of mistaking a systematic artefact of a theory for a 
genuine discovery. This is not what is asserted in this case. 
The statement is that the separate ‘universe’ (using the careful definition) 
under analysis is, in fact, that of the  abstractions themselves. 
The \emph{crux} of the proof is that, in that particular universe 
(not a physical universe), it is a special case where showing it 
exists abstractly is sufficient to show that it exists (in its own universe). 

For particular scientific theories that predict the existence 
of an entity, the `identification of the correctness of its underlying 
principles' becomes the main result of the experiments conducted upon it. 

\emph{Example:} The Dirac Equation was the sole 
contemporary development predicting the existence of the as-yet 
undiscovered anti-matter \cite{dirac1928}, 
which appeared only in the theory describing 
the relativistic motion of a spin-$1/2$ particle using the $4$-component 
bi-spinors that form a Clifford Algebra on the geometry of the four-dimensional 
Minkowski manifold \cite{dirac1930}. 
The equation has remained unchallenged since its 
experimental confirmation in 1929 \cite{close2009antimatter}. 

For physicists discovering structural features of matter, such as 
the shapes of galaxies, distributions of 
dark matter, phases of matter in the early universe, the unification 
of the fundamental forces into mathematical groups, etc., then 
these mathematical concepts are genuine discoveries, rather than 
inventions, unless one takes a solipsistic viewpoint. These 
mathematical objects exist, but in a way that is different to 
`ordinary' objects in the universe. 

The overall result of this work can thus be summarised by the statement that  
‘abstract objects can be demonstrated to be in a distinct universe from ours, 
as the set of both of them together cannot be defined in a consistent way 
logically.' 
The `fundamental object' discussed in the previous section is thus 
an example of a test-case scenario, especially constructed to demonstrate 
that, in principle, there are \emph{at least some} objects that must 
necessarily be in a universe other than our own, for reasons of logical 
consistency - to tame inescapable pathological inconsistencies in trying to 
incorporate them notionally into our own universe.

The pertinent underlying principle for $V'$ is simply that it be consistent; or 
more specifically, that an extant description of it be consistent: 
\begin{equation}
w_{\ro{princ}}' = Y',\quad\mbox{such that}\quad V' = \PP\circ\A_E(Y'),
\end{equation}
for some consistently definable $Y'$. 
It is reasonable to expect that the elements that comprise $V'$   
\emph{constitute evidence} for the extantness of $V'\in\ca{U}'$. 
This can be checked by determining that the condition for evidence  
is satisfied: 
\begin{equation}
\label{eq:finalcond}
w_{\ro{princ}}' \subseteq w_{\ro{princ}}^{\ro{R},T'}[v_k]. 
\end{equation}
This condition, however, is not satisfied in general, due to the fact 
that $Y'$ must contain an abstraction, $\A_{E'}$, which is not present 
in the general formulation of $w_{\ro{princ}}^{\ro{R},T'}$. This makes sense - 
that the truth of a statement does not necessarily entail an extantness. 
In the specific case above, though, we have considered the fundamental 
type of the general $I$-extantness, encompassing all abstractions of 
a specific form, as described in Section \ref{subsect:ext}. In this special 
case, the extantness required by both sides of Eq.~(\ref{eq:finalcond}) 
is the same, and we are simply required to demonstrate that: 
\begin{align}
\label{eq:finalproof}
w_{\ro{princ}}' &= Y' \subseteq w_{\ro{princ}}^{\ro{R},I'}[v_k],\\
\mbox{for}\,\, V' &= \PP\circ\A_I(Y').
\end{align}
\begin{proof}  
$Y'$ is a set containing $N+1$ abstractions, $\A_{v_k}$, with 
$\A_{N+1}\equiv \A_{I'}$:
\begin{equation}
Y' = \{\A_{v_1},\ldots,\A_{v_N},\A_{I'}\} = \bigcup_{k=0}^N\A_{\ro{set}}\circ 
\hat{\Com}^1_{V',i,k} \A_{v_i}. 
\end{equation}
The right-hand side of Eq.~(\ref{eq:finalproof}) takes the form:
\begin{equation}
w^{\ro{R},I'}_{\ro{princ}} = \bigcup_{k=0}^{N-1}
\!\!\!\!\!\!
\bigcup_{\quad p\in[2,n]\backslash\bigcup_{\pi=0}^k\{p^{(\pi)}\}}
\!\!\!\!\!\!\!\!\!\!\!\!\!\!\!\!
\A_{\ro{set}}\circ\hat{\Com}^{(\sum_{j=1}^{n-p+2}i_j)+1}_{\C^{i_1\ldots i_p\ldots i_n}_W,1,v_{k+1}}c 
\cup w_I.
\end{equation}
The left term can be pared down by choosing $n=2$, $i_1=0$ and $i_2=1$, 
and the right term, $w_I$, by choosing $n=1$, $i_1=1$:
\begin{align*}
&\Rightarrow \Big(\bigcup_{k=1}^N \A_{\ro{set}}\circ\hat{\Com}^2_{\PP\Om^1(W),1,v_{k+1}}c
\Big)\,\,\,\cup \,\,
\Big(\A_{\ro{set}}\circ\hat{\Com}^2_{\Om^1(W),m+1,I'}c\Big) \\
&= \{\PP\A_{v_1}(w),\ldots,\PP\A_{v_N}(w)\}\cup\{\A_{I'}(w)\} \\
&= \bigcup_{k=0}^N\A_{\ro{set}}\circ 
\hat{\Com}^1_{V',i,k} \A_{v_i}, \quad\mbox{for}\,\,\A_{N+1}\equiv \A_{I'}.
\end{align*}
\end{proof} 
The final line follows from the labelling principle, $\PP\circ\PP = \PP$, 
and the fact that each $\A_{v_k}$ is in extant form, $\PP\circ\A_I(\A_{y_k})$. 

Note that this proof demonstrates that, in principle, there are objects 
that exist in a universe different from our own, due to the constraints 
of logical consistency. The nature of the test-case object is quite 
rudimentary, but it serves as an initial example. 
The formalism as a whole, though, has been especially fruitful in producing 
the unanticipated results of immediate relevance to the scientific community; 
in particular, a possible explanation for the `unreasonable effectiveness' 
of mathematics is put forward. This represents an important 
development in the philosophical understanding of physics. 
Furthermore, the clarification of \emph{Ansatz}: the process of 
applying an hypothesis to the physical world, and then testing it against 
experiment, represents the main achievement of this work. 
The introduction of a framework within which one can interrogate 
the nature of how hypotheses are applied represents a long-term ambition, 
seemingly missing from the current literature,  
which future research can develop and refine. 

\vspace{-3mm}\section{Conclusion}
This manuscript has attempted to address key issues in physics from the 
point of view of philosophy. By adopting a metaphysical framework closely aligned with that used in the practice of physics, the 
philosophical tool of \textit{Ansatz} was examined, and the process of its use 
was clarified. In this context, a mathematical 
formalism for describing intrinsically non-mathematical 
objects was expounded. In examining the consistency of such a framework, 
a careful distinction between existence (in the mathematical sense) 
and `extantness' (in the sense of phenomena existing in the universe) was 
made. Using the formalism, a general condition for extantness was derived 
in terms of a generalised object, which incorporates the salient 
features of abstraction and projection to the non-mathematical world 
in a way easily manipulated. In principle, the formalism 
may make verifiable predictions, since properties (and the consequences 
of combinations of properties in the form of theorems) can be arranged 
to make strict statements about the behaviour or nature of a system 
or other general objects. 

As an example, a possible explanation of Wigner's `unreasonable 
effectiveness' of mathematics was derived. This demonstrates the ability 
of a metaphysical framework to address important mysteries inaccessible 
from within physics itself. 

Lastly, an attempt was made to classify other universes in a general fashion, 
and to clarify the characteristics and role of evidence for theories that
 provide at least a partial description of a universe.  
The connection between phenomena that constitute evidence  
and the theory itself was established in a proposed Duality Theorem.  
Instead of focusing on attempting an 
\textit{ad-hoc} identification of extra-universal phenomena 
from experiment, the formalism was used 
to \emph{derive} 
basic 
properties of objects that do not align with our universe. 
As a first example toward such a goal, a fundamental object was 
identified, which satisfies the necessary properties for evidence, and 
whose extantness does not coincide with our universe. 
This paves the way for future investigations into more precise details of  
the  
properties of objects and methods amenable to this type of formal inquiry. 

\section*{Declarations}

\textbf{Availability of data and material}: NA. \\
\textbf{Competing interests}: None. \\
\textbf{Funding}: None. \\
\textbf{Authors' contributions}: Jonathan M. M. Hall is the sole author of all work presenting in this manuscript. \\
\textbf{Acknowledgements}: I wish to thank Prof. Derek Abbott for his encouragement in completing this project despite many years without funding. 
I also wish to thank Mr John Lindsay for affording me the opportunity to speak publicly on this work Third Thursday Thinkers, held at the Naval, Military \& Air Force Club of South Australia in 2022.

\bibliographystyle{apalike} 
\bibliography{dualref} 

\end{document}